%% file: ed-consistency-arxiv.tex
\documentclass[runningheads]{llncs}
\usepackage[T1]{fontenc}

\usepackage{graphicx}

\usepackage[textsize=small]{todonotes}
\newcounter{todocounter}

\usepackage{extarrows}

\usepackage{enumitem}
\setlist{nosep}

\input{macros}

\input{tikz-message-functions}

\title{Checking Consistency of Event-driven Traces}

\titlerunning{Checking Consistency of Event-driven Traces}

\author{Parosh Aziz Abdulla\orcidID{0000-0001-6832-6611} \and
Mohamed Faouzi Atig\orcidID{0000-0001-8229-3481} \and
R. Govind\thanks{Corresponding author}\orcidID{0000-0002-1634-5893}\and
Samuel Grahn\orcidID{0009-0004-1762-8061}
\and
Ramanathan S. Thinniyam\orcidID{0000-0002-9926-0931}}

\authorrunning{P. Abdulla, M. Atig, R. Govind, S. Grahn, R. S. Thinniyam}

\institute{Uppsala University, Sweden \email{\{parosh.abdulla,mohamed\_faouzi.atig,govind.rajanbabu,\\samuel.grahn,ramanathan.s.thinniyam\}@it.uu.se}}

\begin{document}

\maketitle              
  
\input{abstract}

\input{intro}

\input{model}

\input{hardness-bounded-handler}

\input{procedure}

\input{experiments}

\input{conclusion}

\clearpage

\bibliographystyle{splncs04}
\bibliography{bibdatabase}
\appendix

\clearpage

\input{app-model}

\clearpage

\input{app-semantics-eq}

\input{app-np-hardness-bddhandler-queue}

\input{appendix-polytimeProcedureQueues}

\input{app-experiments}

\end{document}

%% file: macros.tex
\usepackage[algoruled,vlined,linesnumbered,titlenotnumbered,noend]{algorithm2e}
\usepackage{amsmath}
\usepackage{microtype}
\usepackage{amsfonts}
\usepackage{float}
\usepackage{calc}
\usepackage{wrapfig}
\usepackage{tikz}
\usetikzlibrary{automata,tikzmark,decorations.pathreplacing,petri,backgrounds,positioning,calc,trees,fit,arrows,shapes}
\tikzset{>=triangle 45}

\usepackage{graphicx}
\usepackage{caption}
\usepackage[normalem]{ulem}
\usepackage{hyperref}
\usepackage{tabularx}

\usepackage{bussproofs}
\usepackage{stackengine}
\usepackage{comment}
\usepackage{sidecap}
\usepackage[capitalize]{cleveref}
\usepackage{adjustbox}
\usepackage{pgfplots}
\pgfplotsset{compat=1.13}
\usepackage{pgfplotstable}
\usepackage{tabularx}
\usepackage{multicol}
\usepackage{multirow}
\usepackage{listings}
\usepackage{mathrsfs}

\usepackage{xspace}

\lstdefinestyle{tinyc}{
    basicstyle=\scriptsize\ttfamily,
    keywordstyle=\color{blue}
  }

  \lstdefinestyle{normalc}{
    basicstyle=\ttfamily,
    numbers=none,
    keywordstyle=\color{blue}
  }

  \lstdefinestyle{inlinec}{
    basicstyle=\ttfamily
  }

\newcommand{\NP}{\text{\sc NP}}

\newcommand{\bjnote}[1]{}
\newcommand{\bjforget}[1]{}
\newcommand{\hasbeenremoved}[1]{}

\newcommand\CDSChecker{\textsc{CDSChecker}\xspace}
\newcommand\GenMC{\textsc{GenMC}\xspace}
\newcommand\RCMC{\textsc{RCMC}\xspace}
\newcommand\Nidhugg{\textsc{Nidhugg}\xspace}
\newcommand\Droidracer{\textsc{Droidracer}\xspace}

\newcommand{\true}{\ensuremath{\mathtt{true}}}
\newcommand{\false}{\ensuremath{\mathtt{false}}}

\newcommand{\post}{\texttt{post}}

\newcommand{\po}{\ensuremath{\color{red}\mathtt{po}\color{black}}}
\newcommand{\mo}{\ensuremath{\color{red}\mathtt{mo}\color{black}}}
\newcommand{\rf}{\ensuremath{\color{green}\mathtt{rf}\color{black}}}
\newcommand{\co}{\ensuremath{\color{orange}\mathtt{co}\color{black}}}
\newcommand{\fr}{\ensuremath{\color{RoyalBlue}\mathtt{fr}\color{black}}}
\newcommand{\pb}{\ensuremath{\color{violet}\mathtt{pb}\color{black}}}

\newcommand{\hb}{\ensuremath{\color{blue}\mathtt{hb}\color{black}}}

\tikzset{event/.style={draw=none,fill=none,align=left}}
\tikzset{arc/.style={->,>=stealth,thick}}
\tikzset{arcDerived/.style={->,>=stealth,thick,dashed}}

\makeatletter
\newcommand{\iffref}[3]{\@ifundefined{r@#1}{#3}{#2}}
\makeatother

\def \rf{\mathsf{rf}}
\def \po{\mathsf{po}}
\def \eo{\mathsf{eo}}
\def \co{\mathsf{co}}
\def \pb{\mathsf{pb}}
\def \mo{\mathsf{mo}}
\def \fr{\mathsf{fr}}
\def \hb{\mathsf{hb}}
\def \dto{\mathsf{qo}}
\def \eodag{\mathsf{eo}^{ \dagger}}

\def\mP{\mathcal{P}}

\def\ttx{\texttt{x}}

\def\sharevar{X}

\def\get{\texttt{get}}
\def\post{\texttt{post}}
\def\val{\texttt{val}}
\def\true{\texttt{true}}
\def\false{\texttt{false}}

\def \line{\texttt{line}}
\def \succ{\texttt{succ}}
\def \inst{\texttt{inst}}

\newcommand{\tuple}[1]{\langle#1\rangle}

\def\msR{\mathscr{R}}
\def\msW{\mathscr{W}}
\def\msRx{\msR_{\ttx}}
\def\msWx{\msW_{\ttx}}

\def\adt{\mathsf{MB}}

\def\init{\mathsf{init}}

\def\dsdomain{\mathscr{B}}

\def\rels{\mathsf{rels}}

%% file: tikz-message-functions.tex
\newcommand{\nodeLen}{0.5cm} 
\newcommand{\nodeWidth}{0.5cm} 
\newcommand{\nodeGap}{2cm} 
\newcommand{\nodehGap}{3cm} 
\newcommand{\nodevGap}{1.5cm} 

\newcommand{\ewriteEvent}[5]{
    \node[align=left] (#1) [minimum width=\nodeWidth, minimum height=\nodeLen] 
        at (#2) {$#1=$\\$ \langle #3, \text{W}, #4, #5 \rangle$};
}

\newcommand{\ereadEvent}[4]{
    \node[align=left] (#1) [minimum width=\nodeWidth, minimum height=\nodeLen] 
        at (#2) {$#1=$\\$\langle #3, \text{R}, #4 \rangle$};
}
\newcommand{\postEvent}[4]{
    \node (#1) [minimum width=\nodeWidth, minimum height=\nodeLen] 
        at (#2) {$\langle #3, \text{post}, #4 \rangle$};
}

\newcommand{\epostEvent}[4]{
    \node[align=left] (#1) [minimum width=\nodeWidth, minimum height=\nodeLen] 
        at (#2) {$#1=$\\ $\langle #3, \text{post}, #4 \rangle$};
}

\newcommand{\getEvent}[3]{
    \node (#1) [minimum width=\nodeWidth, minimum height=\nodeLen] 
        at (#2) {$\langle #3, \text{get} \rangle$};
}

\newcommand{\egetEvent}[3]{
    \node[align=left] (#1) [minimum width=\nodeWidth, minimum height=\nodeLen] 
        at (#2) {$#1=$ \\ $\langle #3, \text{get} \rangle$};
}

\newcommand{\drawPo}[2]{
    \draw[->] (#1.east) -- (#2.west) node[midway, above] {$\po$};
}
\newcommand{\vdrawPo}[2]{
    \draw[->] (#1.south) -- (#2.north) node[midway, right] {$\po$};
}

\newcommand{\drawPb}[2]{
    \draw[->] (#1.east) -- (#2.west) node[midway, above] {$\pb$};
}


%% file: abstract.tex
 
\begin{abstract}
  Event-driven programming is a popular paradigm where the flow of execution is controlled by two features: (1) shared memory and (2) sending and receiving of messages between multiple \emph{ handler threads} (just called handler). Each handler has a mailbox (modelled as a queue) for receiving messages, with the constraint that the handler processes its messages sequentially. Executions of messages by different handlers may be interleaved.   
  A central problem in this setting is checking whether a candidate execution is \emph{consistent} with the semantics of event-driven programs.  In this paper, we propose an axiomatic semantics for event-driven programs based on the standard notion of \emph{traces} (also known as execution graphs). We prove the equivalence of axiomatic and operational semantics. This allows us to rephrase the consistency problem axiomatically, resulting in the \emph{event-driven consistency problem}: checking whether a given trace is consistent. We analyze the computational complexity of this problem and show that it is NP-complete, even when the number of handler threads is bounded. We  then identify a tractable fragment: in the absence of nested posting, where handlers do not post new messages while processing a message, consistency checking can be performed in polynomial time. Finally, we implement our approach in a prototype tool and report on experimental results on a wide range of  benchmarks.

  \keywords{Event-driven programs \and Consistency-checking \and Verification.}
\end{abstract}

%% file: intro.tex
\section{Introduction}

Event-Driven (ED) programming has emerged as a powerful paradigm for building scalable and responsive systems capable of handling a large number of user interactions concurrently~\cite{CunninghamK05,libevent,libasync,mace,GayLBWBC03,HillSWHCP00,apple,microsoft}. It is widely used across various domains, including file systems~\cite{Mazieres01}, high-performance servers~\cite{Dabek:event-driven-02}, systems programming~\cite{P:pldi13}, and smartphone applications~\cite{mednieks2012programming}.
Event-driven programs have become so common that they are  considered a core topic under \emph{Programming Fundamentals} according to IEEE and ACM computing curricula \cite{EDpedagogy2021}.  
ED programming extends multi-threaded shared-memory programming through the use of messages, thus using both shared-memory as well as message-passing.

Verification of ED programs, in addition to the usual challenges associated with shared-memory multi-threaded program verification,  has to deal with the non-determinism introduced by the sending and receiving  of messages between multiple \emph{ handler threads} (just called handler). Each handler has a mailbox (modelled as FIFO queue, following ~\cite{maiyaPartialOrderReduction2016,Event-DrivenSMC-OOPSLA-15}) for receiving messages, with the constraint that the handler processes its messages sequentially. Executions of messages by different handlers may be interleaved.  
A well-established technique for verifying multi-threaded programs is stateless model checking (SMC)\cite{Godefroid:popl97}, which has proven effective for detecting concurrency bugs. SMC has been implemented in several tools - including VeriSoft\cite{Godefroid:verisoft-journal}, \textsc{Chess}\cite{MQBBNN:chess}, Concuerror\cite{Concuerror:ICST13}, \Nidhugg~\cite{DBLP:journals/acta/AbdullaAAJLS17}, rInspect~\cite{DBLP:conf/pldi/ZhangKW15}, \CDSChecker~\cite{NoDe:toplas16}, \RCMC~\cite{KLSV:popl18}, and \GenMC~\cite{GenMC-CAV-21} - and applied to realistic programs~\cite{GoHaJa:heartbeat,KoSa:spin17}.
To efficiently explore execution traces, SMC tools often employ dynamic partial order reduction (DPOR)~\cite{Valmari:reduced:state-space,Peled:representatives,CGMP:partialorder,Godefroid:thesis,DBLP:journals/pacmpl/AbdullaAJN18,DBLP:journals/acta/AbdullaAAJLS17,DBLP:conf/cav/AbdullaAJL16,DBLP:conf/popl/AbdullaAJS14}. DPOR avoids redundant exploration by recognizing and pruning equivalent executions. DPOR does this by exploring the space of all \emph{traces} (also called execution graphs~\cite{KLSV:popl18}).  Intuitively, a trace is a summary of the important concurrency information contained in a program execution, represented as a directed graph whose edges are  the union of certain relations (defined further below). 

A central component of DPOR  techniques is \emph{consistency checking} (e.g., see Section 5 of~\cite{Abdulla2019} and Section 4 of~\cite{Kokologiannakis2022}), which involves determining whether a candidate trace is realizable, i.e., whether there exists an execution of the program that respects all the relations implied by the trace. The consistency checking problem has been extensively studied on its own for different programming models, notably by Gibbons and Korach (for Sequential Consistency)~\cite{GibbonsK97} from 1997 and continued in several works (e.g., ~\cite{Bouajjani17,Tunc2023,Chakraborty2024}).

We refer to the consistency checking in the event-driven setting as the event-driven sequential consistency problem.
In this work, we  study the event-driven sequential consistency problem.  We first propose an axiomatic semantics for ED programs. We then establish the equivalence between the operational and axiomatic semantics of ED programs.
Next, we  explore the complexity landscape of the event-driven sequential consistency problem.

Concretely, we consider as input a trace represented by a set of events and relations among them. The goal is to determine whether this trace can arise from a valid event-driven execution. These relations include \emph{Program Order} (the order in which instructions are fetched from the code associated with the message), \emph{Read-From} relation (which relates each read to the write that it reads from), \emph{Coherence Order} (which specifies the order between  writes on the same variable).
The above relations already exist for general multi-threaded shared-memory programs. In addition, our traces contain the  \emph{Execution Order}   (which fixes the order in which the messages of the same handlers are executed), \emph{Message order} (ordering the posting of messages to the same handler) and \emph{Posted-by} (relating the instruction posting the message to the instruction starting the execution of the message). We  prove that event-driven consistency problem is NP-complete, \emph{even when the number of handler threads is bounded}. On the positive
side, we identify a tractable fragment: in the absence of nested posting - where handlers do not post new messages while processing a message - consistency checking can be performed in polynomial time. Finally, we
implement our approach in a prototype tool and report on experimental
results on a wide range of  event-driven benchmarks, both synthetic and from real-world event-driven programs.
\vspace*{0.2in}

\noindent\textbf{Related work.} Race detection in event-driven programs has been studied~\cite{raychevEffectiveRaceDetection2013,Maiya:pldi14}. There has also been work on partial order reduction in this setting~\cite{maiyaPartialOrderReduction2016,ATVA2023} and stateless model checking~\cite{Event-DrivenSMC-OOPSLA-15}. The paper~\cite{ATVA2023}  considers the consistency problem in the case of mailboxes modeled as multisets, showing its {NP}-hardness. In the specialised setting of ED programs for real-time systems, the work ~\cite{gantyAnalyzingRealTimeEventDriven2009} shows that checking  safety properties is undecidable. The robustness problem, which asks if a given an ED program has the same behaviour as if it were to be run on a single thread, has been studied in ~\cite{bouajjaniVerifyingRobustnessEventDriven2017}. Several efforts have been made to provide language support for ED programming such as Tasks \cite{fischerTasksLanguageSupport2007} as well as the P programming language ~\cite{P:pldi13}. In particular, P programming was built to provide safe asynchronous ED programming from the ground up and used to implement and validate the USB driver in Windows 8.  Finally, the consistency problem has been extensively studied for different programming languages (e.g., \cite{GibbonsK97,Bouajjani17,Tunc2023,Chakraborty2024,ATVA2023}). 
However, as far as we know, this is first time that the consistency problem is studied in the context of ED programs with FIFO queues as mailboxes. 

%% file: model.tex
\section{Event-driven programs: Syntax and Semantics}~\label{sec:model}
In the following,  we will first  give the syntax of Event-driven (ED) programs. Then, we will  describe the  operational semantics of ED programs. Next,  we  will define the notion of traces of ED programs and give an equivalent axiomatic definition of traces. Finally,  we will  define the event-driven consistency problem.

\subsection{Syntax of Event-Driven Programs} 
\label{sub:syntax_of_event_driven_programs}

\begin{wrapfigure}[8]{R}{0.65\textwidth} 
  \centering      
  \vspace{-\baselineskip}
  \vspace{-0.5cm}
  \begin{adjustbox}{width=0.64\textwidth}
    \begin{tikzpicture}
    \node at (1,6) {$<prog> ::= \textbf{vars} \ <var>^{*} \  \textbf{handlers} \ <handler>^{*} \  \textbf{msgs} <msg>^*$};
    \node at (1,5.5) {$<handler>::= \ <handlerId> \textbf{regs} \  <reg>^{*} $};
    \node at (1.5,5) {$<msg> ::= <msg name> \ <inst>^* \ <label>: <last>$};
    \node at (1.5,4.5) {$<inst> ::= <label> : \ <stmt> \ $};
    \node at (1,4) {$<stmt> ::= <var> = <reg> | <reg> = <var> | <reg> = <exp>$};
    \node at (1.25,3.5) {$if <cond> \textbf{goto} :<label> | \textbf{goto}: <label> $};
    \node at (0.5,3) {$post(<handlerId>,<msg name>) $};  
  \end{tikzpicture}
  \end{adjustbox}
  \caption{Syntax of Event-Driven Programs}
  \label{fig:progSyntax}
\vspace{-\baselineskip}
\end{wrapfigure}

The syntax of  event-driven programs we consider is shown in \cref{fig:progSyntax}. An event-driven program $\mP$ has a finite set $H$ of \emph{handlers}\footnote{ED  programs often have designated handler threads with mailboxes as well as \emph{non-handler} threads which do not have an associated mailbox. However, we can think of a non-handler thread as a handler to which a message is never posted and hence simplify notation by assuming that all threads are handlers.}, each $h \in H$ having a finite set $R_h$ of \emph{local registers}. We denote $R=\bigcup_h R_h$. The handlers interact via a finite set $X$ of \emph{(shared) variables}, as well as via a finite set $M_h $ of \emph{messages} which are posted to the mailbox $b_h$ associated with each handler $h$. We assume that the local registers and shared variables take  values from a data domain $D$. The message sets of different handlers are assumed to be disjoint with $M=\bigcup_h M_h$ the set of all messages. 
Each message has a \emph{message name} and comprises of a sequence of instructions, ending in a special instruction $last$ which indicates the end of the message.

Each instruction consists of a unique \emph{label} followed by a \emph{statement}. A statement of the form $<var> = <reg>$ in the grammar indicates the writing of a register value into a shared variable. A statement of the form $<reg>=<var>$ indicates the read operation of a shared variable which is then stored in the local register of a handler. More complex manipulations of data domain values are assumed to be performed within handlers through the use of \emph{expressions} as in $<reg>=<exp>$. These expressions are assumed to only use  local registers. The $\textbf{goto}$ statement moves the program control to the indicated label, with conditional branching allowed using the $if$ construct.The condition $cond$ in the $if$ statement uses only the local registers of the handler which executes the statement. The post statement posts a message to the mailbox of the indicated handler. We assume that the labels in $\mP$ form a finite set $L$ and there is a successor function $\succ \colon L  \mapsto L$ which indicates the flow of program control. Each label $l \in L$ has an associated instruction $\inst(l)$ which is given by the function $\inst$. 

\subsection{Operational Semantics of Event-Driven Programs} 

\label{sub:semantics_of_event_driven_programs}

We now describe the operational semantics of ED programs, focusing on how handlers interact with mailboxes during the execution of an event-driven program.

\smallskip

\noindent \textbf{Handler.}
A handler $h$ repeatedly extracts a message from its mailbox, executes the code of the message to completion, then extracts another message and executes its code, and so on. This extraction is modelled as a $\get$ event.
We use a counter at each handler in order to generate unique message IDs. 
Note that execution of messages by different handlers could be interleaved.  
Further, while executing the code of a message by a handler, messages could be added to its mailbox.
The execution of a message is done one instruction at a time. 
At any point of time, a handler has at most one \emph{active message} which is being executed. An underlying \emph{nondeterministic scheduler} decides which handler to run at a step.

\smallskip

\noindent \textbf{Mailbox.}
A mailbox is a labelled transition system $\adt = \tuple{\dsdomain$, $\beta_\init$, $\{\get,\post\}$, $\Sigma$, $\rightarrow}$, where $\dsdomain$ is the set of configurations of $\adt$, $\beta_\init \in \dsdomain$ is the initial configuration, and $\Sigma$  is the set of messages (including a special symbol $\bot$). We assume that the operations that can be performed on $\adt$ are $\{ \get,\post\}$ and the transition relation $\rightarrow \subseteq \dsdomain \times \{ \get,\post\}\times \Sigma \times \dsdomain $ specifies the semantics of the operations. 
 In this paper, the operations are of two kinds: $\get$ which downloads a message from the mailbox, and $\post$ which adds a new message into the mailbox.   Since the mailbox is modelled as a FIFO queue and follows the first-in-first-out semantics, we have  $\dsdomain=\Sigma^*$ and $\beta_\init=\varepsilon$. We write $\beta \xrightarrow{o,\sigma}\beta'$ to denote that the message $\sigma$ is returned by (resp. posted by) the operation $o$ if it is a $\get$ (resp. $\post$) and $\beta=\beta' \cdot \sigma$ (resp. $\beta'=\sigma \cdot \beta'$), while transitioning from configuration $\beta$ to configuration $\beta'$. 
A \emph{run} $\rho=\beta_\init \xrightarrow{o_1, \sigma_1}\beta_1 \xrightarrow{o_2, \sigma_2} \ldots \beta_n$ of $\adt$ is a finite sequence of transitions starting from the initial configuration $\beta_\init$.

 \smallskip

\noindent \textbf{Configuration of ED-programs.}
 We use $h,g, \ldots$ for handlers, $m$ for messages, x,y,z for shared variables, and a,b,c for local registers. Members of $D$ will be denoted by $v$. Configurations of programs and mailboxes will be denoted by $\alpha$ and $\beta$ respectively. 
 The local state $s_h=\tuple{\val,\beta,\line,mid,mcount}$ of a handler $h$ is a tuple containing the valuation $\val \colon R_h \mapsto D$ of its local registers, the configuration of its mailbox $\beta$, $\line$ which is the label of the next instruction that will be executed by the handler, the message id $mid$ of its currently active message and a counter $mcount$. The valuation function $\val$ is extended to expressions in the standard way. When a message $m$ is to be posted by handler $h$ to the mailbox of handler $h'$, the local counter $mcount$ is incremented by 1. A unique mid  $(h,mcount)$ is generated and associated with the message instance. 
We will write $s_h.\val, s_h.\beta$ etc to denote the components of $s_h$. Given a particular message name $m$, let $m.\init$ denote the label of  the first instruction in the code of  $m$. 

A configuration $\alpha=(\{s_h \mid h \in H\},\nu)$ of a program consists of the local state of each handler $h$ along with a valuation $\nu \colon X \mapsto D$ of the shared variables. We sometimes write $\alpha.\nu$, $\alpha.s_h$ etc to denote the valuation of global variables and the local state of handler $h$ respectively, in configuration $\alpha$. 
A program $\mP$ starts in some initial configuration\footnote{Note that we intentionally do not specify the initial values of local registers and  shared variables, or the initial message each handler should execute, as the event-driven consistency problem treats the program code as a black box. However, our framework  can be easily extended to take into account  such initial conditions.} $\alpha_0=(\{s_h^0 \mid h \in H\},\nu)$ which satisfies the condition that for each handler $h$, we have $s_h^0.\line=m.\init$ for some $m \in M_h$ i.e. the execution of some message is initialised in each handler. Note that there is no post event associated with this initialization. Furthermore, the mailboxes are all empty i.e. $s_h^0.\beta=\beta_{\init}$ for all $h$, $s_h^0.mid=(h,0)$  and $s_h^0.mcount=1$.

\begin{figure}[h]
  \vspace{-\baselineskip}
\flushleft\textsc{Write}\\
\AxiomC{$\inst(\alpha.s_h.\line)=l_i \colon x=a \wedge \alpha.s_h.\val(a)=v$}
\UnaryInfC{$\alpha \xrightarrow{\tuple{h,write,x,v}} \alpha(\nu \leftarrow \nu(x\leftarrow v), s_h.\line \leftarrow \succ(s_h.\line))$ }
\DisplayProof
\flushleft\textsc{Post}\\
\AxiomC{$\alpha.s_h.\line=l\colon post(h',m)$\;\;$\alpha.s_{h'}.\beta \xrightarrow{\post,(m,newmid)} \beta'$\;\;$newmid=(h,\alpha.s_h.mcount)$}
\UnaryInfC{$\alpha \xrightarrow{\tuple{h,\post,h',newmid}} \alpha(s_{h'}.\beta \leftarrow \beta',s_h.mcount \leftarrow s_h.mcount +1, s_h.\line \leftarrow \succ(s_h.\line))$ }
\DisplayProof
 \flushleft\textsc{Read}\\
  \AxiomC{$\inst(\alpha.s_h.\line)=l\colon a=x$}
  \UnaryInfC{$\alpha \xrightarrow{\tuple{h,read,x}} \alpha(s_h.\val \leftarrow s_h.\val(a\leftarrow \nu(x)), s_h.\line \leftarrow \succ(s_h.\line))$ }
  \DisplayProof
\flushleft\textsc{Get}\\
\AxiomC{$\alpha.s_h.\line=l \colon last$ \;\;$\alpha.s_h.\beta_h \xrightarrow{\get,(m,mid)}_\adt \beta'$}
\UnaryInfC{$\alpha \xrightarrow{\tuple{h,\get,mid}} \alpha(s_h.\beta \leftarrow \beta', s_h.mid \leftarrow mid,s_h.\line \leftarrow m.\init)$ }
\DisplayProof
  \caption{A subset of  transition rules of ED  programs with $\alpha=(\{s_h \mid h \in H\},\nu)$. }
  \label{fig:transRules}
\end{figure}

A transition $\alpha \xrightarrow{a} \alpha'$, between two configurations $\alpha$ and $\alpha'$, occurs on either the execution of an instruction or a $\get$ operation. The subset of rules dictating transitions relevant to concurrency is shown in \cref{fig:transRules}.  The other rules i.e. local transitions within a thread, can be found in \cref{app:figtransRules} in Appendix~\ref{sec:app-model}. Given a  tuple/mapping $f$, we use $f(x \leftarrow d)$ to denote the tuple/mapping $f'$ which agrees with $f$ on all parameters except $x$, on which it takes the value $d$. We write $f(x \leftarrow d)$ (resp. $f(x_1 \leftarrow d_1, \ldots, x_n \leftarrow d_n)$) instead of  $f( g \leftarrow g(x \leftarrow d))$  (resp. $f(x_1 \leftarrow d_1) \cdots  (x_n \leftarrow d_n)$) when it is clear from the context.

An \emph{execution sequence} or run $\rho$ of program $\mP$ is a finite sequence of transitions $\alpha_0 \xrightarrow{a_1} \alpha_1 \xrightarrow{a_2} \ldots \xrightarrow{a_n} \alpha_n$ 
starting with an initial configuration $\alpha_0$.

\subsection{Events, Traces and Axiomatic Consistency} 
\label{sub:axiomatic_consistency}
In this subsection, we introduce an axiomatic semantics for ED programs. We first define the relevant types of events and then formalize the notion of a trace.

\paragraph{Events} 
\label{par:events}
An event is a collection of information about a transition that is meant to be made visible. Transitions which contain such information are called event-transitions (observe that local-transitions  do not have corresponding events). As shown in \cref{fig:transRules}, there are four types of event-transitions: reads, writes, posts and gets. The event is obtained from the event-transition by dropping the mid and newmid information. Note that newmid is basically a newly created mid written this way for clarity. Formally, 

\begin{itemize}
    \item A $\mathsf{write}$ event is a tuple $e=\tuple{h,write,x,v}$ which denotes the writing of the value $v$ by handler $h$ into global variable $x$. We say $e.var=x$ and $e.val=v$. We denote by $\msWx$ the set of all write events on $x$.
    
    \item A $\mathsf{read}$ event is a tuple $e=\tuple{h,read,x}$ which denotes the reading of the value stored in global variable $x$ by handler $h$. We say $e.var=x$. We denote by $\msRx$ the set of all read events on $x$.
    
    \item A $\post$ event is a tuple $e=\tuple{h,post,h'}$ which denotes the posting of a message by handler $h$ to the mailbox of the handler $h'$. We write $e.sender$ to denote $h$ and $e.receiver$ to denote $h'$.

    \item A $\get$ event is a tuple $e=\tuple{h,\get}$ which denotes the downloading of a message by handler $h$. 
\end{itemize}

In general, we write $e.h$ to denote the handler on which a message is being executed. In particular, $e.h$ is the same as $e.sender$ for a post event. Given a transition $\alpha \xrightarrow{a} \alpha'$ we write $e(a)$ for the event corresponding to $a$ if $a$ is an event-transition. In case the transitions are indexed e.g. $a_i$ then we just write $e_i$ instead of $e(a_i)$. For an event $e$, we denote by $e.type$ the \emph{type} of the event i.e. whether it is a read, write, $\post$ or $\get$. 
Note that unless necessary, we omit handler identifiers from events for readability.

\paragraph*{Traces} 
  Let $\rels = \{ \rf, \co, \po, \eo,  \pb, \mo\}$ be a set of relation names. A trace is a directed graph $\tau =(E, \Delta)$ where $E$ is a finite set of events, $\Delta \subseteq E \times \rels \times E$ is a set of edges on $E$ with labels from $\rels$. 
Let $E_h=\{e \mid e.h =h\}$ be the  set of events occurring on handler $h$. Let $G_h=\{e \mid e.type=get\} \cap E_h$ and $P_h=\{e \mid e.type =post \wedge e.receiver=h\}$ be respectively the get events of  handler $h$ and the post events to  $h$.
  The following  conditions are satisfied by $\Delta$: 
\begin{enumerate}
  \item[$\rf$:] (reads-from) maps each read instruction to a write instruction. 
  For each $x \in \sharevar$ and each $e \in \msRx$, there exists exactly one $e' \in \msWx$ such that $e'\; \rf \; e$.

  \item[ $\po$:] (program order) is 
   a union of total orders on the set of events $E_h$ which occur on a particular handler.
    This is a total ordering on all events which happen as part of the execution of a particular message instance.  
Formally,  we have
\begin{itemize}
\item For any $e \in E_h\setminus G_h$,  there exists at most one event $e' \in G_h$ s.t. $e' \; \po \; e$. 

\item For every  $e, e'' \in E_h\setminus G_h$ such that $e' \; \po \; e$ and $e' \; \po \; e''$ for some $e' \in G_h$, it is the case that either $e \;  \po \; e''$ or $e'' \; \po  \; e$. 

\item Let $E'_h$ be the set of events $e \in E_h\setminus G_h$ such that there is no event $e' \in G_h$ with $e' \; \po \; e$. Then, $\po$ is a total order over $E'_h$. Furthermore, for every $e \in E'_h$  and  $e' \in G_h$, we have  $e \; \po \; e'$. This means that all  events of the initial  message are ordered before the events of any other message.
   \end{itemize}
   
  \item[ $\co$:] (coherence order) For each pair of writes $e,e' \in \msWx$, either $e \;\co\; e'$ or $e'\; \co \;e$. 
  Note that $\co$ is a total order on the set $\msWx$ for each $x \in \sharevar$.

  \item[$\eo$:] (execution order)  is 
  a total order on the set of get events occurring on a handler. Let $e,e' \in G_h$ for some $h$. Then either $e \; \eo  \; e'$ or $e' \;  \eo \; e$.

 \item [ $\pb$:] (posted by) is 
 a relation which relates each get event on a handler to the corresponding post event. In other words, it is a bijection between the sets $P_h$ and $G_h$: for each $e \in G_h$ there is exactly one $e'  \in P_h$ such that $e' \; \pb \; e$. 
  \item[ $\mo$:] (message order)
  orders the events that posts messages to the mailbox of a  particular handler. For every $e,e' \in P_h$ either $e \ \mo \ e'$ or $e' \ \mo \ e$. 

\end{enumerate}
 A partial trace is a subgraph of a  trace.  
A partial trace $\tau'=(E', \Delta')$ is said to \emph{extend} a partial trace $\tau =(E, \Delta)$ if $E \subseteq E', \Delta \subseteq \Delta'$. A linearization $\pi=(E,\leq_\pi)$ of a partial trace $\tau=(E,\Delta)$ is a total ordering $\leq_\pi$ satisfying $\delta \in \Delta \Rightarrow \delta \in \leq_\pi$.

\paragraph*{Traces of Programs} Given a program $\mP$ and its execution $\rho =\alpha_0 \xrightarrow{a_1} \alpha_1 \xrightarrow{a_2} \ldots \xrightarrow{a_n} \alpha_n$, we define the set $E(\rho)=\{ e_i \mid  \text{  $a_i$ is an event-transition} \}$ to be the \emph{event set of $\rho$}. Clearly $\rho$ induces a total order $\leq_{\rho}$ on $E(\rho)$ defined in the natural way: $e_i \leq_{\rho} e_j$ iff $i \leq j$. 
Each $\get$ event-transition specifies an mid for the message instance which is obtained from the mailbox. The execution of this message instance may contain more event-transitions later in $\rho$. Hence we extend the notion of mid to all non-$\get$ event-transitions in the following way. 
For each event $e_i$ which is a not a get event, let $e_j$ be the first  preceding get event $e_j$ in the order $\rho$ such that $e_i.h=e_j.h$, if such an event exists. We  assign the message id $a_j.mid$ to the transition $a_i$, since by the event-driven semantics, only one message can be executed by a handler at any point in time. If no such  get event exists, then we assign message id $(h,0)$ to $a_i$. Note that a post event has both an mid from the message it is part of as well as a newmid for the message it is creating.

Recall that for  $x \in \sharevar$, we have $\msR_{x}=\{ e \in E \mid e.type=read, e.var=x \}$, $\msW_{x}=\{ e \in E \mid e.type=write, e.var=x \}$. The event set $E$ together with the total order $\leq_\rho$ derived from a run \emph{induces} a  trace $\tau(\rho)$ in the following way:
\begin{enumerate}
\item[$\rf$:] If $e_i \; \leq_\rho e_j$ where $e_i \in \msWx, e_j \in \msRx$, and for all $e_i \leq_\rho e_k \leq_\rho e_j$ we have $ e_k \not \in \msWx$, then $e_i \; \rf \; e_j$.
\item[$\co$:] If $e_i,e_j \in \msWx$ for some $x$ and $e_i \; \leq_\rho \; e_j$ then $e_i \; \co \; e_j$.
  \item[$\po$:]If $e_i,e_j$ are such that $a_i.mid=a_j.mid$ and $e_i \; \leq_\rho \; e_j$ then $e_i \; \po \; e_j$.
  Further, if $e_i,e_j$ are such that $e_i.h=e_j.h$, $a_i.mid=(h,0)$ and $a_j.mid \neq a_i.mid$  then $e_i \; \po \; e_j$.
  \item[$\eo$:] If  $e_i \; \leq_\rho \; e_j$ satisfies $e_i.type=e_j.type=\get$ and $e_i.h=e_j.h, a_i.mid \neq a_j.mid$ then $e_i \;\eo \; e_j$. 
  \item[$\pb$:] If $e_i \; \leq_\rho \; e_j$ satisfies $e_i.type=post, e_j.type=get$ and $a_i.newmid=a_j.mid$ then $e_i \; \pb \; e_j$ 
  \item[$\mo$:]If $e_i,e_j$ are both post events such that $e_i.receiver=e_j.receiver$ and $e_i \; \leq_\rho \; e_j$ then $e_i \; \mo \; e_j$.
\end{enumerate}
The definition of $\rf$ ensures that every read on a variable reads from the latest write on that variable. The coherence order $\co$ is just the sequence of writes on a variable. The mid information tells us which set of instructions belong to the same message from which we can infer $\po$. Similarly the mid information also tells us which $\get$ is posted-by which $\post$. The sequence of message executions on a handler ($\eo$) and the order in which messages were posted to a handler ($\mo$) can be inferred from $\leq_\rho$. This gives us the following lemma: 

\begin{lemma}
\label{lem:progTraceAbstract}
  For any program $\mP$ and its execution $\rho$, $\tau(\rho)$ is a trace.
\end{lemma}

\begin{remark}
 Recall that a handler processes a message in its entirety before accessing the next message from its mailbox. Hence all events belonging to one message of $h$ must occur before all events of 
 all subsequent messages of $h$. Further, note that since the mailbox is a FIFO queue, then another requirement is that the order in which messages are removed from the mailbox needs to respect the order in which they are added to the mailbox.
 In this case, $\mo$ is a total ordering on all events that post messages to the same handler.   
The restriction here is that the order in which messages are extracted should be according to the $\mo$ between the events that posts the messages to the same mailbox. 
\end{remark}

\paragraph*{Axiomatic Consistency}
We introduce conditions under which a  trace $\tau$ is said to be \emph{axiomatically consistent} and show that this happens iff $\tau$ can be derived from the run of some event driven program. The conditions are given as is standard by means of acyclicity of the union of relations. 
To this end, we introduce new relations: The  \emph{queue order} $\dto$ is defined  as $\dto=\pb^{-1}.\mo.\pb$.
 The from-reads relation be defined as $\fr=\rf^{-1}.\co$. Let $\eodag= (\po^{-1})^*.\eo.\po^*$.

\begin{definition}
  \label{defn:axCons}
  A trace $\tau$ is said to be axiomatically consistent if the happens-before relation $\hb=(\po \; \cup \; \rf  \; \cup \; \fr  \; \cup \; \co  \; \cup \; \pb  \; \cup \; \mo  \; \cup \; \eodag  \; \cup \; \dto )$ is acyclic.
\end{definition}

\begin{theorem}
  \label{thm:consEq}
  A trace $\tau$ is axiomatically consistent iff there exists an event-driven program $\mP$ and a run $\rho$ such that $\tau=\tau(\rho)$. 
\end{theorem}
The proof of the theorem can be found in Appendix \ref{app:sec:axSemanticsEqProof}. 

\subsection{Event-driven consistency problem}

Having defined both operational and axiomatic semantics, we now introduce the event-driven consistency problem, which asks whether a partial trace can be extended to an axiomatically consistent  trace.

We first recall a related problem in the non-event-driven setting, where only the relations $\po$ (program order), $\rf$ (reads from) and $\co$ (coherence order) are relevant. When all three are given, consistency checking is tractable. However, if only $\po$ and $\rf$ are provided (aka $\rf$-consistency), is known to be $\NP$-complete \cite{GibbonsK97}. 
In ED programs, we additionally deal with  $\pb$, $\mo$ and $\eo$. Given the NP-completeness of the rf-consistency problem, it is natural to ask about the complexity of ED-consistency where $mo$ and $eo$ are not provided\footnote{Note that the $\eo$ order should respect the $\mo$ order, since the mailbox is  a FIFO queue. When we talk about partial traces, we do not explicitly mention the relations present. The understanding is that both $\mo$ and  $\eo$ are missing. }.
As in the nonED case, our proof of equivalence in Theorem \ref{thm:consEq} implies that consistency is in polynomial time if all the relations are given. This leads to the well-motivated consistency problem for event-driven programs.
Let $\rels'=\{\po \cup \rf \cup \co  \cup \pb \}$

\begin{definition}[ED-Consistency Problem]
    \label{def:queueConsProblem}
     Given a partial trace $\tau'=(E',\Delta')$ with $\Delta' \subseteq E \times \rels' \times E$, decide whether there exists an axiomatically consistent extension trace  $\tau=(E,\Delta)$ of $\tau'$  such that $\Delta' \cap (E \times \rels' \times E)=\Delta \cap (E \times \rels' \times E)$.
\end{definition}

%% file: hardness-bounded-handler.tex
\section{Complexity of Event-driven Consistency}
\label{sec:hardness}

In this section, we study the complexity of the event-driven consistency problem. 
We show that the problem is $\NP$-hard.  
Further, since our reduction uses only 12 handlers, this also implies the hardness for the more restricted version of the problem, with only a bounded number of handler threads. 

The proof will follow from a reduction from \emph{3-Bounded Instance 3SAT} (3-BI-3SAT in short), a problem known to be {\NP}-complete~\cite{Tovey84}.
Further details and a proof of correctness are given in Appendix~\ref{app:hardness}.

\begin{definition}[3-BI-3SAT] 
\label{def:bddInstance3SATProblem}
       A 3-BI-3SAT  problem  is the Boolean satisfiability problem restricted to conjunctive normal form formulas such that: (1) each clause contain two or  three literals, (2)     each variable occurs in at most 3 clauses, and (3) each variable appears at most  once per clause. 
\end{definition}

\begin{theorem}~\label{thm:NP-hardness-bddhandler-queue}
        The ED-consistency problem is NP-complete for traces with at most 12 handlers. 
\end{theorem}
The proof is done by reduction from the  3-BI-3SAT.
Let $\phi$ be a 3-BI-3SAT instance with variables $x_1,x_2,\ldots,x_n$ and clauses $C_1,C_2, \ldots, C_m$. We will construct a partial ED trace $\tau=(E,\Delta)$, with $\Delta \subseteq E \times \rels' \times E$, such that $\tau$ can be extended to a axiomatically consistent  trace  $\tau'=(E,\Delta')$, with  $\Delta' \cap (E \times \rels' \times E)=\Delta \cap (E \times \rels' \times E)$, iff $\phi$ is satisfiable.

\noindent\textbf{High level structure. }
The construction of the trace $\tau$ is divided into two stages, which we call Stage 1 and Stage 2 respectively. There are 8 handlers in Stage 1 and 5 handlers in Stage 2. One handler, namely $h_W$ is common to both stages, hence totally there are 12 handlers. If a satisfying assignment exists for  $\phi$, then there is a program which can execute the events in Stage 1 followed by those in Stage 2, i.e., $\tau$ is consistent. If $\phi$ is unsatisfiable then there is no witnessing execution possible which executes both stages and $\tau$ is inconsistent. 

\underline{Stage 1} corresponds to the selection of a satisfying assignment $f$ for $\phi$.  
 We can encode the information of whether a variable $x_i$ is assigned true or false using the order of execution of two messages $m_{i,1}$ and $m_{i,0}$ on the same handler, where $x_i$ is assigned true (resp. false) if $m_{i,1}$ (resp. $m_{i,0}$) is executed later. Unfortunately, this will not work due to technical difficulties faced in clause verification (see Remark \ref{rem:3copies} in Subsection \ref{sub:stage1}). 
This necessitates our extremely technical reduction which makes use of the structure of the 3-BI-3SAT instance where each variable occurs in \emph{at most} three clauses and the variables occurring in a clause are all different. We have to create (at most) 3 copies of the messages, one for each clause in which $x_i$ occurs and find a way to synchronise the assignment between these three copies. 

Hence the messages for $x_i$ are actually of the form $m_{i,j,b}$ where $j$ refers to clause $C_j$ and $b \in {0,1}$. Using the technique of \emph{post sequences} which use nested posting (explained using example in Subsection \ref{sub:stage1}), 
we post the set $M$ of $m_{i,j,b}$  messages in the queue of $h_W$ in some order $\sigma$. The remaining 7 handlers of Stage 1 are used to shuffle the messages in $M$ with certain restrictions on the order $\sigma$ of messages. 

 The set $S$ of all the possible orders $\sigma$ is such that, every $\sigma$ is constrained to be \emph{consistent} (see Challenge 2 of Subsection \ref{sub:stage2}) 
 with some particular assignment $f$ of variables of $\phi$. There are no other constraints on the order $\sigma$. At the end of Stage 1, the queues of all other Stage 1 handlers is empty and the queue of $h_W$ is populated in some order $\sigma \in S$ consistent with some assignment $f$. Note that there are multiple $\sigma$ which are consistent with a particular $f$, this fact will be important later.

\underline{Stage 2} Let us fix $\sigma$ and $f$ from Stage 1. Then Stage 2 verifies that $f$ indeed satisfies all the clauses of $\phi$. For this, we build a clause gadget $G_j$ corresponding to each clause $C_j$. The set $E_G$ of events of these clause gadgets occupy the 4 non-$h_W$ handlers of Stage 2. The $E_G$ events belong to an initial message of each of the 4 handlers, and consist purely of read and write events. Recall that the queue of $h_W$ is populated at this point with message set $M$. The information regarding the assignment is encoded in the order of the messages in $h_W$. This information is transferred to the other 4 non-$h_W$ handlers via a technique we call \emph{sandwiching} (explained using example in Subsection \ref{sub:stage2}). 
There are now two possibilities:\\ 
 \noindent(1) If $f$ is not a satisfying assignment, then some clause $C_j$ is not satisfied by $f$. In this case, any order $\sigma$ of messages consistent with $f$ will induce a $\hb$ (happens before) cycle in the corresponding gadget $G_j$ via the sandwiching. Therefore Stage 2 cannot be executed by any witnessing execution. If there are no satisfying assignments, then $\phi$ is unsatisfiable and hence $\tau$ is not consistent.  

 \noindent(2) If $f$ is a satisfying assignment then there is some order $\sigma$ of the messages in $M$ which is consistent with $f$ such that there is a witnessing execution. The clause gadgets are executed interleaved with the messages in $M$ due to the sandwiching. The execution happens sequentially i.e. $G_1$ is executed, then $G_2$ etc. This implies $\tau$ is consistent.

\begin{remark}
    When we say that message $m$ is posted to handler $h$ before message $m'$ is posted to $h$, we refer to the order of post operations as executed in the witnessing execution.
\end{remark}

Through the use of examples, we now give intuition on how the two Stages work, beginning with Stage 2.
\subsection{Stage 2: Checking satisfaction of clauses.}
\label{sub:stage2}
We assume that $h_W$ has been populated with messages in accordance with a variable assignment function $f$ in Stage 1. 
Each clause $C_j$ is associated with a \emph{clause gadget} $G_j$ implemented using handlers $h_{C_a}, h_{C_b}, h_{C_c}, h_{C_d}$ (the non-$h_W$ handlers of Stage 2). For example, for $C_2 = x_1 \vee x_2 \vee \overline{x_n}$, Figure~\ref{fig:clauseGadget} shows the corresponding gadget $G_2$ and two messages posted to $h_W$ in Stage 1.
For reasons of space we use $W$ and $R$ for $write$ and $read$ to describe   events.

Within each gadget, there are \emph{boxes} that contain a read followed by a write event in $\po$. Each box is linked to messages in $h_W$ through $\rf$ relations. 
Consider the box $b_1$ in Figure~\ref{fig:clauseGadget}. Here, the variable $\overline{l^2_{n,1}}$ in $e_{10}$ has the information: superscript $2$ for clause $C_2$, subscript $n,1$ indicating the literal $\overline{x_n}$ and $1$ indicating it is the first event in the box. The events in $b_1$ are linked to the read and write events in the message $m_{n,2,0}$ via $\rf$ arrows. The direction of the arrows implies that the events in box $b_1$ have to be executed after event $e_{15}$ and before $e_{16}$ which are both in message  $m_{n,2,0}$. This is the technique we call \emph{sandwiching}.

Similarly $b_2$ has to be executed during the execution of $m_{n,2,1}$. Suppose $m_{n,2,0}$ is executed before $m_{n,2,1}$ as indicated by the $\eo$, this means that $x_n$ is assigned the value $\true$. This sandwiching induces the red $\hb$ relation shown between $e_{11}$ and $e_{12}$.

\noindent \underline{Clause satisfaction.} Notice that similar boxes can be drawn around events $e_2,e_3$ and $e_4,e_5$ corresponding to copying the assignment to variable $x_1$ and for $e_6,e_7$ and $e_8,e_9$ for variable $x_2$. Each of these boxes has similar sandwiching $\rf$ relations to messages in $h_W$ which are not shown in the figure. The three red $\hb$ arrows correspond to setting each of the three variables in $C_2$ to a value that falsifies the corresponding literal in $C_2$. The events $e_1$ and $e_{14}$ use a variable $z_2$ (where the subscript refers to the clause $C_2$) and are connected by an $\rf$. Under these conditions, a cycle is formed and thus the clause gadget cannot be executed. On the other hand, if even one of the red  arrows is flipped (indicating that a literal of $C_2$ is set to $\true$), then the arrows form a partial order allowing execution of the clause gadget $G_2$. 

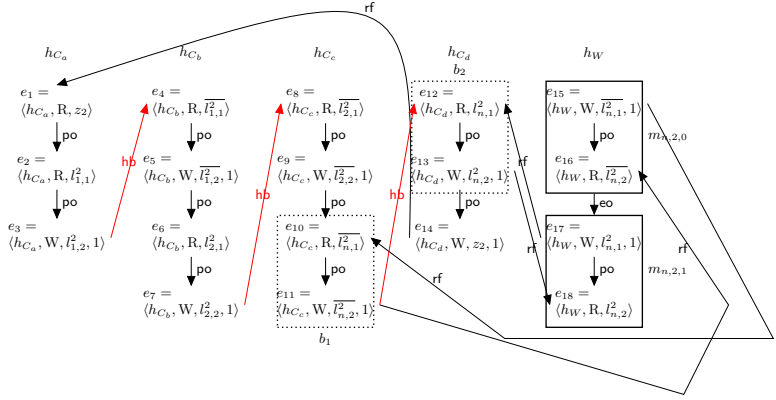
\begin{figure}[h]
        \vspace{-\baselineskip}
        \centering
        \resizebox{0.85\textwidth}{!}{
\begin{tikzpicture}
\foreach \i in {1,...,5} {
    \coordinate (H\i) at ({(\i-1) * \nodehGap},{-0.25* \nodevGap}); 
}

\node (hCa) [minimum width=\nodeWidth, minimum height=\nodeLen] at (H1) {$h_{C_a}$}; 

\node (hCb) [minimum width=\nodeWidth, minimum height=\nodeLen] at (H2) {$h_{C_b}$}; 

\node (hCc) [minimum width=\nodeWidth, minimum height=\nodeLen] at (H3) {$h_{C_c}$};

\node (hCd) [minimum width=\nodeWidth, minimum height=\nodeLen] at (H4) {$h_{C_d}$};

\node (hW) [minimum width=\nodeWidth, minimum height=\nodeLen] at (H5) {$h_W$};
\foreach \i in {1,...,3} {
       \coordinate (E\i) at (0,{-(\i ) * \nodevGap}); 
}
\foreach \i in {4,...,7} {
       \coordinate (E\i) at (\nodehGap,{-(\i - 3) * \nodevGap}); 
}
\foreach \i in {8,...,11} {
       \coordinate (E\i) at (2*\nodehGap,{-(\i - 7) * \nodevGap}); 
}
\foreach \i in {12,...,14} {
       \coordinate (E\i) at (3*\nodehGap,{-(\i - 11) * \nodevGap}); 
}

\foreach \i in {15,...,18} {
       \coordinate (E\i) at (4*\nodehGap,{-(\i - 14) * \nodevGap}); 
}


\ereadEvent{e_1}{E1}{h_{C_a}}{z_2};
\ereadEvent{e_2}{E2}{h_{C_a}}{l^2_{1,1}};
\ewriteEvent{e_3}{E3}{h_{C_a}}{l^2_{1,2}}{1};
\node (b1) [draw,dotted, line width=0.3mm, minimum width=4.3*\nodeWidth,minimum height=5*\nodeLen, label={south:$b_1$}] at ($(E10)!0.5!(E11)$) {};
\node (b2) [draw,dotted, line width=0.3mm, minimum width=4.3*\nodeWidth,minimum height=5*\nodeLen, label={north:$b_2$}] at ($(E12)!0.5!(E13)$) {};

\ereadEvent{e_4}{E4}{h_{C_b}}{\overline{l^2_{1,1}}};
\ewriteEvent{e_5}{E5}{h_{C_b}}{\overline{l^2_{1,2}}}{1};
\ereadEvent{e_6}{E6}{h_{C_b}}{l^2_{2,1}};
\ewriteEvent{e_7}{E7}{h_{C_b}}{l^2_{2,2}}{1};

\ereadEvent{e_8}{E8}{h_{C_c}}{\overline{l^2_{2,1}}};
\ewriteEvent{e_9}{E9}{h_{C_c}}{\overline{l^2_{2,2}}}{1};
\ereadEvent{e_{10}}{E10}{h_{C_c}}{\overline{l^2_{n,1}}};
\ewriteEvent{e_{11}}{E11}{h_{C_c}}{\overline{l^2_{n,2}}}{1};

\ereadEvent{e_{12}}{E12}{h_{C_d}}{l^2_{n,1}};
\ewriteEvent{e_{13}}{E13}{h_{C_d}}{l^2_{n,2}}{1};
\ewriteEvent{e_{14}}{E14}{h_{C_d}}{z_2}{1};

\ewriteEvent{e_{15}}{E15}{h_W}{\overline{l^2_{n,1}}}{1};
\ereadEvent{e_{16}}{E16}{h_W}{\overline{l^2_{n,2}}};
\ewriteEvent{e_{17}}{E17}{h_W}{l^2_{n,1}}{1};
\ereadEvent{e_{18}}{E18}{h_W}{l^2_{n,2}};

\node (m1) [draw, line width=0.3mm, minimum width=4.3*\nodeWidth,minimum height=5*\nodeLen, label={east:$m_{n,2,0}$}] at ($(E15)!0.5!(E16)$) {};
\node (m2) [draw, line width=0.3mm, minimum width=4.3*\nodeWidth,minimum height=5*\nodeLen, label={east:$m_{n,2,1}$}] at ($(E17)!0.5!(E18)$) {};


\vdrawPo{e_1}{e_2};
\vdrawPo{e_2}{e_3};

\vdrawPo{e_4}{e_5};
\vdrawPo{e_5}{e_6};
\vdrawPo{e_6}{e_7};

\vdrawPo{e_8}{e_9};
\vdrawPo{e_9}{e_{10}};
\vdrawPo{e_{10}}{e_{11}};

\vdrawPo{e_{12}}{e_{13}};
\vdrawPo{e_{13}}{e_{14}};

\vdrawPo{e_{15}}{e_{16}};
\vdrawPo{e_{17}}{e_{18}};

\draw[->] (e_{14}.west) .. controls (8,1.5) .. node[midway, above] {$\rf$} (e_1.north);

\draw[->] ($(e_{11}.east)$) -- (14,-8) -- (15,-6)-- (e_{16}.east) node[midway, below] {$\rf$};
\draw[->] (e_{15}.east) --(16,-6.75) --(10,-6.75) -- ($(e_{10}.east)$) node[midway, above] {$\rf$};

\draw[->] ($(e_{13}.east)$) -- (e_{18}.west) node[midway, below] {$\rf$};
\draw[->] (e_{17}.west) -- ($(e_{12}.east)$) node[midway, above] {$\rf$};
\draw[->] (m1.south) -- (m2.north) node[midway, right] {$\eo$};

\draw[->,color=red] (e_{11}.east) -- ($(e_{12}.west)$) node[midway, above] {$\hb$};
\draw[->,color=red] (e_7.east) -- ($(e_8.west)$) node[midway, above] {$\hb$};
\draw[->,color=red] (e_3.east) -- ($(e_4.west)$) node[midway, above] {$\hb$};
\end{tikzpicture}
}
\caption{The structure of clause gadget $G_2$ of $C_2$ for the example in Figure \ref{fig:varsAndClauses}.}
\label{fig:clauseGadget}
\end{figure}

Note that the clause gadgets $G_1,G_1,\ldots, G_m$ are placed in that order in the handlers $h_{C_a}, h_{C_b}, h_{C_c}, h_{C_d}$ and connected by $\po$ arrows. For example, $G_j.e_3$ will be $\po$ before $G_{j+1}.e_1$ in $h_{C_a}$, $G_j.e_7$ will be $\po$ before $G_{j+1}.e_4$ in $h_{C_b}$, etc.  In other words, the events of each of these four handlers can be assumed to be in an initial message in the respective handlers. There is no posting of events either from or to these 4 handlers. 

\subsection{Stage 1: Encoding variable assignments.}
\label{sub:stage1}
In this Stage, we use the handlers $h_V,h_{t_1},h_{t_2},h_{t_3},h_{t_4}, h_{t_5},h_{t_6}$ in order to post messages to $h_W$. 
We stated that an assignment to variable $x_i$ can be encoded as the order between  $m_{i,j,0},m_{i,j,1}$.
\vspace{0.1in}

\noindent \underline{Challenge 1:} How can we ensure that  the messages $m_{i,j,0},m_{i,j,1}$ can be posted in any order to $h_W$? \\ 
In order to solve this, we use \emph{nested posting}. 
$h_V$ posts $m'_{i,j,0}$ to $h_1$ and $m'_{i,j,1}$ to $h_2$, which in turn post $m_{i,j,0}$ and $m_{i,j,1}$ to $h_W$. Since $m'_{i,j,0}$ and $m'_{i,j,1}$ are on different handlers,
they can be executed in any order, thus ensuring that $m_{i,j,0},m_{i,j,1}$ can posted in any order to $h_W$. Next we take up the reason for using multiple messages for each variable. 

\begin{remark}
\label{rem:3copies}
  Consider the sandwiching technique that we presented in Stage 2 in order to copy the assignment of a variable to the clause gadget. Suppose we were to use a single pair of messages $m_{i,0},m_{i,1}$ in $h_W$ for a variable $x_i$ from which this value was copied to the different clauses in which $x_i$ occurs. This means that any handler in which a clause gadget is being executed would be blocked from running till all of the clauses containing $x_i$ are able to finish executing the boxes corresponding to $x_i$. This leads to a cascading set of blocked handlers, requiring an unbounded number of handlers to execute the clause gadgets. In order to overcome this difficulty, we have to use upto three copies of the two messages $m_{i,0},m_{i,1}$ as mentioned before. But this leads to a different challenge.      
\end{remark}
\vspace{0.1in}

\noindent \underline{Challenge 2:} How can we ensure that the different copies of the messages corresponding to the same variable exhibit the same value?\\ 
\vspace{0.05in}

To solve this, we rely on the structure of the 3-BI-3SAT formula. Figure~\ref{fig:varsAndClauses} depicts a grid with variables as rows and clauses as columns. Marked cells indicate variable occurrences. For each literal $l$, we define a \emph{post sequence} $p^l$ composed of segments $p^l_1, \ldots, p^l_7$ corresponding to marked and unmarked cells.

\noindent To address this, we further extend the nesting of posts, relying on the structure of the 3-BI-3SAT formula $\phi$. 
Figure~\ref{fig:varsAndClauses} depicts a grid with variables as rows and clauses as columns. Marked cells indicate variable occurrences.
Each row  contains  2 or 3 marked cells and each column contains 4 or 6 marked cells as per the restriction on 3-BI-3SAT.
For each literal $l$, we define a \emph{post sequence} $p^l$ composed of segments $p^l_1, \ldots, p^l_7$ corresponding to marked and unmarked cells.
Marked cells trigger the posting of messages to $h_W$, while unmarked ones simply post to $h_V$ and defer execution. The post sequence of each marked cell is designed to enforce consistent ordering of the associated messages.
We will now describe the post sequences in detail.

A post sequence is a partial trace of the form
\begin{tikzpicture}
        \centering
        \foreach \i in {1,...,5} {
       \coordinate (E\i) at ({(\i - 1) * \nodeGap}, 0); 
}

\postEvent{e1}{E1}{h_0}{h_1};
\getEvent{e2}{E2}{h_1};

\postEvent{e3}{E3}{h_1}{h_2};
\node (e4) [minimum width=\nodeWidth, minimum height=\nodeLen] 
        at (E4) {$\ldots$};
\postEvent{e5}{E5}{h_{n-1}}{h_n};

\drawPb{e1}{e2};
\drawPo{e2}{e3};
\drawPb{e3}{e4};
\drawPo{e4}{e5};
\end{tikzpicture}

We will simply write this as $p=\tuple{h_1,\post,h_2,\post,\ldots,\post,h_n}$. In case $h_i=h_{i+1}=\ldots=h_j$ we will further shorten this to $\tuple{h_1,\post$,$h_2,\post,\ldots$, $h_i$, $\post^{j-i-1},h_j$, $\post,h_{j+1}$, $\post,\ldots,\post,h_n}$.

\begin{wrapfigure}[22]{r}{0.43\textwidth}
        \vspace{-2\baselineskip}
                \centering
                \begin{tikzpicture}[scale=0.5]
        
                    \def\cellSize{0.5cm}
        
                    \foreach \i in {1,...,8} {
                        \foreach \j in {1,...,8} {
                            \draw[thick] (\j-1, -\i+1) rectangle (\j, -\i);
        
                            \coordinate (mid-\i-\j) at (\j-0.5, -\i+0.5);
                        }
                    }
        
                \node (x1) [minimum width=\nodeWidth, minimum height=\nodeLen] 
                        at ($(mid-1-1)-(1,0)$) {$x_1$};
                \node (barx1) [minimum width=\nodeWidth, minimum height=\nodeLen] 
                        at ($(mid-2-1)-(1,0)$) {$\overline{x_1}$};
                \node (x2) [minimum width=\nodeWidth, minimum height=\nodeLen] 
                        at ($(mid-3-1)-(1,0)$) {$x_2$};
                \node (barx2) [minimum width=\nodeWidth, minimum height=\nodeLen] 
                        at ($(mid-4-1)-(1,0)$) {$\overline{x_2}$};
                \node (ldots1) [minimum width=\nodeWidth, minimum height=\nodeLen] 
                        at ($(mid-5-1)-(1,0)$) {$\vdots$};
                \node (xn) [minimum width=\nodeWidth, minimum height=\nodeLen] 
                        at ($(mid-7-1)-(1,0)$) {$x_n$};
                \node (barxn) [minimum width=\nodeWidth, minimum height=\nodeLen] 
                        at ($(mid-8-1)-(1,0)$) {$\overline{x_n}$};
        
                \node (C1) [minimum width=\nodeWidth, minimum height=\nodeLen] 
                        at ($(mid-1-1)+(0,1)$) {$C_1$};
                \node (C2) [minimum width=\nodeWidth, minimum height=\nodeLen] 
                        at ($(mid-1-2)+(0,1)$) {$C_2$};
                \node (C3) [minimum width=\nodeWidth, minimum height=\nodeLen] 
                        at ($(mid-1-3)+(0,1)$) {$C_3$};
                \node (C4) [minimum width=\nodeWidth, minimum height=\nodeLen] 
                        at ($(mid-1-4)+(0,1)$) {$C_4$};
                \node (ldots1) [minimum width=\nodeWidth, minimum height=\nodeLen] 
                        at ($(mid-1-5)+(0,1)$) {$\ldots$};
                \node (C8) [minimum width=\nodeWidth, minimum height=\nodeLen] 
                        at ($(mid-1-6)+(0,1)$) {$C_8$};
                \node (ldots2) [minimum width=\nodeWidth, minimum height=\nodeLen] 
                        at ($(mid-1-7)+(0,1)$) {$\ldots$};
                \node (Cm) [minimum width=\nodeWidth, minimum height=\nodeLen] 
                        at ($(mid-1-8)+(0,1)$) {$C_m$};

                \node (n11) [minimum width=\nodeWidth, minimum height=\nodeLen] 
                        at ($(mid-1-2)$) {\tiny{$(1,1)$}};
                \node (n21) [minimum width=\nodeWidth, minimum height=\nodeLen] 
                        at ($(mid-2-2)$) {\tiny{$\overline{(1,1)}$}};
        
                \node (n12) [minimum width=\nodeWidth, minimum height=\nodeLen] 
                        at ($(mid-1-4)$) {\tiny{$\overline{(1,2)}$}};
                \node (n22) [minimum width=\nodeWidth, minimum height=\nodeLen] 
                        at ($(mid-2-4)$) {\tiny{$(1,2)$}};
        
                \node (n13) [minimum width=\nodeWidth, minimum height=\nodeLen] 
                        at ($(mid-1-6)$) {\tiny{$(1,3)$}};
                \node (n23) [minimum width=\nodeWidth, minimum height=\nodeLen] 
                        at ($(mid-2-6)$) {\tiny{$\overline{(1,3)}$}};
        
                \node (n31) [minimum width=\nodeWidth, minimum height=\nodeLen] 
                        at ($(mid-3-1)$) {\tiny{$\overline{(1,1)}$}};
                \node (n41) [minimum width=\nodeWidth, minimum height=\nodeLen] 
                        at ($(mid-4-1)$) {\tiny{$(1,1)$}};
        
                \node (n32) [minimum width=\nodeWidth, minimum height=\nodeLen] 
                        at ($(mid-3-2)$) {\tiny{$(2,2)$}};
                \node (n42) [minimum width=\nodeWidth, minimum height=\nodeLen] 
                        at ($(mid-4-2)$) {\tiny{$\overline{(2,2)}$}};
        
                \node (n33) [minimum width=\nodeWidth, minimum height=\nodeLen] 
                        at ($(mid-3-3)$) {\tiny{$\overline{(3,1)}$}};
                \node (n43) [minimum width=\nodeWidth, minimum height=\nodeLen] 
                        at ($(mid-4-3)$) {\tiny{$(3,1)$}};
        
                \node (n71) [minimum width=\nodeWidth, minimum height=\nodeLen] 
                        at ($(mid-7-2)$) {\tiny{$\overline{(1,3)}$}};
                \node (n81) [minimum width=\nodeWidth, minimum height=\nodeLen] 
                        at ($(mid-8-2)$) {\tiny{$(1,3)$}};
        
                \foreach \i in {5,...,6} {
                        \foreach \j in {1,...,8} {
                            \node (n\i\j) [minimum width=\nodeWidth, minimum height=\nodeLen] 
                        at ($(mid-\i-\j)+(0,0.25)$) {\tiny{$\vdots$}};
        
                        }
                    }
                \foreach \i in {1,...,8} {
                        \foreach \j in {5,7} {
                            \node (m\i\j) [minimum width=\nodeWidth, minimum height=\nodeLen] 
                        at ($(mid-\i-\j)$) {\tiny{$\ldots$}};
        
                        }
                    }
                \end{tikzpicture}
                \caption{\footnotesize{Relationship between variables and clauses dictating the nesting of posts. Empty cell means variable does not occur in clause (not all nonempty cells are shown). A cell marked $(u,v)$ or  $\overline{(u,v)}$ indicates that it is the $u$-th occurrence of the variable in a clause and is the $v$-th variable of the clause.
                The bars on the tuple indicate the polarity of the variable occurrence. $C_2=x_1 \vee x_2 \vee \overline{x_n}$, $x_1$ occurs in $C_2,C_8$ and $\overline{x_1}$ occurs in $C_4$.}}
                \label{fig:varsAndClauses}
        \vspace{-\baselineskip}
\end{wrapfigure}
Consider the row labelled by $\overline{x_1}$ in the Figure \ref{fig:varsAndClauses}. The post sequence is the concatenation of 7 post sequences $p^{\overline{x_1}}=p^{\overline{x_1}}_1p^{\overline{x_1}}_2p^{\overline{x_1}}_3p^{\overline{x_1}}_4p^{\overline{x_1}}_5p^{\overline{x_1}}_6p^{\overline{x_1}}_7$ where $p^{\overline{x_1}}_2,p^{\overline{x_1}}_4,p^{\overline{x_1}}_6$ correspond to the cells marked $\overline{1,1}$, $1,2$ and $\overline{1,3}$ respectively, while the others correspond to the part of the row consisting of unmarked cells, with $p^{\overline{x_1}}_1$ for the part from the beginning till the first marked cell, etc. Each of the post sequences $p^{\overline{x_1}}_1$, $p^{\overline{x_1}}_3$, $p^{\overline{x_1}}_5$, $p^{\overline{x_1}}_7$ consists of a long sequence of posts to $h_V$ of length the number of unmarked cells in the segment they correspond to. For example $p^{\overline{x_1}}_1=p^{\overline{x_1}}_3=\tuple{h_V,\post,h_V}$ while $p^{\overline{x_1}}_5=\tuple{h_V,\post^3,h_V}$ since there are 3 empty cells in between (in the figure they are not explicitly shown, but rather by $\ldots$, but once can infer that the boxes corresponding to $C_5,C_6,C_7$ are empty along this row). The  idea is that the post sequences are executed column by column. The empty cell post sequences simply `send to back of queue' while the marked cells are responsible for populating $h_W$ with an appropriate sequence of messages as explained below.

We now describe the post sequences made in the marked cells. Consider the 6 marked cells corresponding to column $C_2$. 
Top to bottom, these are $p^{x_1}_2, p^{\overline{x_1}}_2, p^{x_2}_4, p^{\overline{x_2}}_4, p^{x_n}_2, p^{\overline{x_n}}_2$. 
Let us consider the post sequence for a cell labelled $(u,v)$ (resp. $\overline{(u,v)}$), indicating that it is the $u$-th occurrence of the variable in a clause and is the $v$-th variable of the clause, with the bar indicating whether the variable or its negation occurs in the clause. Suppose $u \neq 1$, then the post sequence is $\tuple{h_V,\post,h_{t_v},\post,h_V}$ for both $(u,v)$ as well as $\overline{(u,v)}$. If $u=1$ then the post sequence is $\tuple{h_V,\post^2,h_{t_v},\post,h_V}$ for $(u,v)$ but it is $\tuple{h_V,\post,h_{t_{v+3}},\post,h_{t_{v}},\post,h_V}$ for $\overline{(u,v)}$. 
For example,  $p^{x_n}_2$ which is marked $\overline{(1,3)}$ has the post sequence $\tuple{h_V,\post,h_{t_6},\post,h_{t_3},\post,h_V}$. Intuitively, the post sequences of the variable and its negation move to different handlers $h_{t_k}$ before coming back to the same handler iff a variable is occurring for the first time in a clause i.e., if $u=1$.    
We modify each post sequence of a marked cell to post a message $m_{i,j,b}$ (corresponding to the occurrence of $x_i$ in $C_j$ in positive or negative form based on the value of the bit $b$) to $h_W$ just before its return to $h_V$. For example, in $p^{x_n}_2$ we insert the events $e_2,e_3,e_4,e_5$ between the events $e_1$ and $e_6$ which are part of $p^{x_n}_2$ as follows:

\begin{tikzpicture}
        \foreach \i in {1,...,4} {
       \coordinate (E\i) at ({(\i-1)  * \nodehGap},0); 
}
\coordinate (E5) at ({3  * \nodehGap},-\nodevGap);
\coordinate (E6) at ({1  * \nodehGap},-\nodevGap);
\coordinate (E7) at ({2  * \nodehGap},-\nodevGap);

\egetEvent{e_1}{E1}{h_{t_3}};
\epostEvent{e_2}{E2}{h_{t_3}}{h_W};
\egetEvent{e_3}{E3}{h_W};
\ewriteEvent{e_4}{E4}{h_W}{l^2_{n,1}}{1};
\ereadEvent{e_5}{E5}{h_W}{l^2_{n,2}};
\epostEvent{e_6}{E6}{h_{t_3}}{h_V};
\node[draw] (e_7) [minimum width=\nodeWidth, minimum height=\nodeLen] 
        at (E7) {$ Stage\ 2$}; 
\drawPo{e_1}{e_2};
\drawPb{e_2}{e_3};
\drawPo{e_3}{e_4};
\vdrawPo{e_4}{e_5};
\vdrawPo{e_2}{e_6};
\draw[dotted,->] (e_4.south) -- (e_7.north) node[midway, right] {$\rf$};
\draw[dotted,->] (e_7.east) -- (e_5.west) node[midway, right] {$\rf$};
\end{tikzpicture}
This means that six messages are posted to $h_W$ corresponding to $C_2$. The assignment to $x_1$ and $x_n$ are chosen by using different handlers $h_{t_i}$ for them, but the assignment to $x_2$ was already chosen when executing the post sequence for $C_1$ corresponding to $x_2$. Hence $p^{x_2}_4, p^{\overline{x_2}}_4$ will both contain $h_{t_2}$ and the corresponding messages will be posted to $h_W$ in the order already chosen. Crucially, we prevent orderings in $h_W$ which do not correspond to consistent assignment of variables. However we allow all other possible reorderings of messages and this is essential for the verification in Stage 2, where only some of these reorderings may be allowed based on the partial order of events in a satisfiable clause i.e. one where not all red  $\hb$ arrows are present (see Figure \ref{fig:clauseGadget}). 

\vspace{0.1in}

In this way, the post sequences ensure that:
\begin{enumerate}
        \item all copies of the same variable agree on the assignment,
        \item messages are posted in an order reflecting this assignment,
        \item the clause gadgets can detect satisfiability by checking the consistency of the induced trace $\tau$.
\end{enumerate}

We provide more details of the reduction and a formal proof of correctness in 
Appendix~\ref{app:hardness}.

%% file: procedure.tex
\section{Event-driven programs without nested posting}

We have shown that the event-driven consistency problem is NP-hard, even when the number of handlers is bounded. However, a closer examination of our reductions reveals that they rely critically on the ability of a message to post another message to a handler—a feature we refer to as  \emph{nested posting}. This observation motivates the study of a restricted setting in which nested posting is disallowed: does the problem become easier in this case?
We answer this question affirmatively in Theorem~\ref{thm:boundedHandlernoNestingQinP} (proof provided in Appendix~\ref{prf:thm:boundedHandlernoNestingQinP}), which shows that the event-driven consistency problem becomes tractable under this restriction. Specifically, we present a polynomial-time procedure for checking consistency when the number of handlers is bounded and nested posting is not allowed.

\begin{theorem}
	\label{thm:boundedHandlernoNestingQinP}
	Given a trace $\tau=(E,\Delta')$ containing $k$ handlers and no nesting of posts, the  ED-Consistency problem for $\tau$ can be solved in time polynomial in $|E|=n$ and exponential in the number  $k$ of handlers. 
\end{theorem}

\noindent \textbf{Proof Sketch.}
Since there is no nesting of posts, all post events occur in the initial message of each handler. The $\po$ order within the initial message therefore implies an $\mo$ order on all posts made by a given handler $h$. Due to queue semantics, this also implies the corresponding $\eo$ order on the corresponding $\get$ events and also the $\eodag$ relation between all pairs of events occurring in two messages ordered by $\eo$. 
This implies that for the messages posted \emph{to} a handler $h'$, we obtain $k$ different total orders of based on the handler $h$ making the post. We can now express the consistency problem in terms of program termination as follows. We define the notion of a configuration $C$ which consists of 
\begin{enumerate}
	\item [(P1)] $k$ pointers for each handler $h'$ which indicates the messages which have been executed so far in each of the $k$ total orders (for a total of $k^2$ pointers), and 
	\item [(P2)] $k$ additional pointers which indicate the event to be executed next in the message currently being executed in each handler. 
\end{enumerate}
This information requires $O(k^2\log(n))$ space and thus the total number of configurations is exponential in $k$ and polynomial in $n$. 

We can now construct the configuration graph $G$ which consists of nodes representing the configurations and edges which indicate when a transition is possible (e.g. when the a pointer in P1 is moved to the successor event inside a message, all other pointers remaining the same). Thus, the trace is consistent iff there is a path from the initial configuration to the final configuration in  $G$.

\section{Consistency checking: Procedure and Optimisations}
\label{sec:proc}

In this section, we propose two concrete procedures for checking  event-driven consistency.
Both the procedures take as input
the trace as a graph $G$ whose nodes are events and whose edges are
the program relations 
$\rf, \co$, message relation $\po$, and posted-by relation $\pb$. 
The first procedure is based on some \emph{saturation rules}, designed to accelerate consistency checking.
The second procedure involves encoding 
the consistency problem as a constraint satisfaction problem and uses the Z3 SMT solver to solve it.  Observe  that the ED-consistency problem is in $\NP$. Since checking consistency is in polynomial time if all the relations are given, it suffices to guess the missing relations. 

\begin{wrapfigure}[8]{R}{0.58\textwidth} 
	\vspace{-\baselineskip}
		\vspace{-0.5cm}

	\scriptsize	  

	\begin{algorithm}[H]
		\caption{Consistency checking}
		\label{algo:consistency-checking}
		\KwIn{A partial trace $G=(E,\Delta)$ where $\Delta \subseteq E \times \rels' \times E$.}
		\If{$G$ contains a cycle}{
			\Return{Inconsistent} 
		}
		\Else{
				Apply saturation rules (1), (2), and (3) to $G$\;
				$\mathsf{CheckQconsistency(G)}$\;
		}
	\end{algorithm}
\end{wrapfigure}

\subsection{Procedure using Saturation Rules}

The $\mathsf{CheckQconsistency(G)}$ procedure iterates over all possible assignments of $\eo$ and $\mo$ edges for $G$, and for each of these assignments, (1) adds the   $\fr$, $\dto$ and $\eodag$  edges that are defined in Section~\ref{sub:axiomatic_consistency}, and checks whether there is one for which $G$ is acyclic after the addition of these edges.
If there is an assignment for which the graph $G$ is acyclic, the procedure returns true, otherwise, it returns false.

\noindent{\bf Saturation rules.}
To reduce the search space and speed up the consistency check, we define the following saturation rules, depicted in Figure~\ref{fig:patterns-heuristic}:

\begin{figure}[!tbp]
	\centering
        \resizebox{0.85\textwidth}{!}{
	\begin{minipage}{.15\textwidth}
		\centering
		\includegraphics[width=2.2cm]{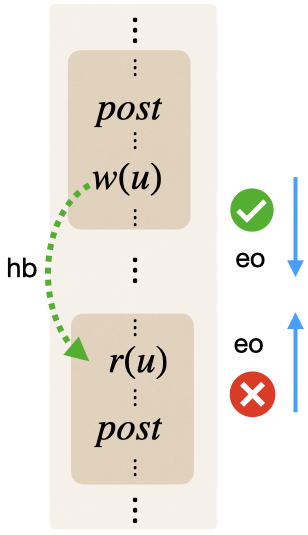}
		\end{minipage}%
	\hspace*{0.4in}
	\begin{minipage}{0.25\textwidth}
		\centering
		\includegraphics[width=3.7cm]{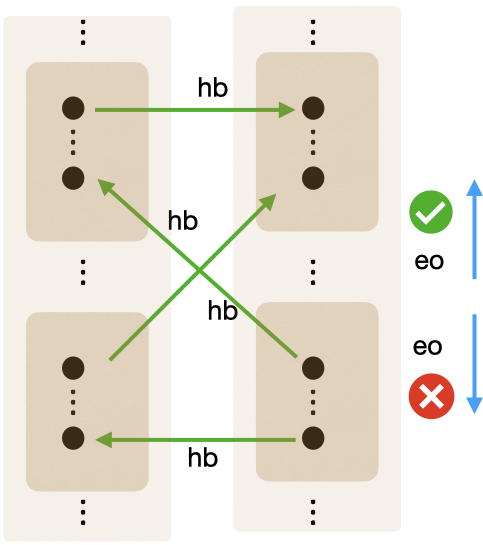}
	\end{minipage}
	\hspace*{0.4in}
	\begin{minipage}{0.25\textwidth}
	  \centering
	  \includegraphics[width=3.7cm]{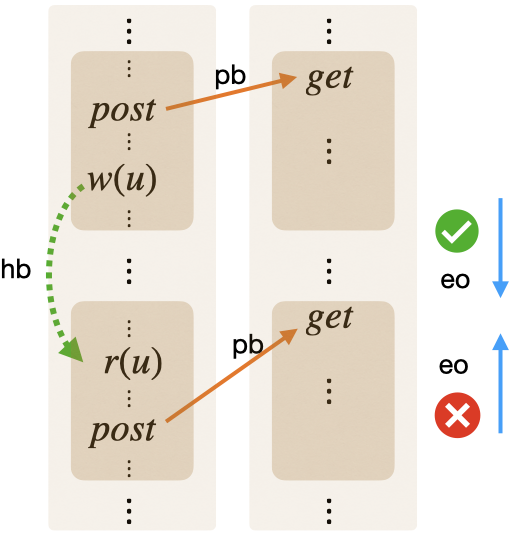}
	\end{minipage}
}
		\vspace{-\baselineskip}
\caption{\footnotesize{Patterns corresponding to saturation rules.} 
		}
	\vspace{-\baselineskip}
	  \label{fig:patterns-heuristic}
  \end{figure}

\noindent\textbf{Rule 1}: If there is a $\hb $ edge between two events in two different messages of the same handler, the $\eo$ edge between these messages have to respect this order.

\noindent\textbf{Rule 2}: Consider the sequence of $\hb$ edges between events as depicted in the second figure. It can be seen that this pattern forces the $\eo$ edge to be as shown in the figure - the inverse $\eo$ edge immediately creates a cycle in the trace.

\noindent\textbf{Rule 3}: This rule checks that the order in which messages are processed follows the queue semantics. Consider the sequence of $\hb$ edges between events as depicted in the third figure. It can be seen that this pattern forces the $\eo$ edge to be as shown in the figure - the inverse $\eo$ edge immediately creates a cycle in the trace.
	
\begin{wrapfigure}[13]{R}{0.5\textwidth}
	\scriptsize
	\begin{algorithm}[H]
		\caption{\small{Consistency checking}}
		\label{algo:z3-checking}
		\KwIn{A partial trace $G=(E,\Delta)$ where $\Delta \subseteq E \times \rels' \times E$.}
        $T \gets Enum(E)$\;
        $O \gets PartialOrder(T)$\;
        $I \gets Z3 instance$\;
        \For{$(a,b) \in \Delta$}{
          $I.assert(O(a, b))$\;
        }
        \For{$h \in handlers$}{
          \For{$m_{1}, m_{2} \in h.messages$}{
            $I.assert(O(m_{1}.last, m_{2}.first) \vee O(m_{2}.last, m_{1}.first))$\;
            $I.assert(O(m_{1}.post, m_{2}.post) \vee O(m_{2}.post, m_{1}.post))$\;
          }
        }
        \Return $I.check()$\;
	\end{algorithm}
\end{wrapfigure}

\subsection{Procedure using SMT encoding}
Here, we present an algorithm to check consistency of a partial event-driven trace using the Z3 solver~\cite{deMoura2008Z3}.
Note that each event $e$ is part of one message, and each message is part of one handler. The algorithm uses the \emph{Special Relations} theory in Z3.

The algorithm takes the input trace $G$ and the set of events $E$. An enum datatype $T$ is created, where each value corresponds to an event.
Then, a partial order $O$ is declared over the events in $T$ using Z3's Special Relations theory - this order will represent the orderings that must be satisfied for the trace to be consistent. Then, a Z3 solver instance $I$ is created to solve the logical constraints defined over the events and their orderings. Then, in lines 4-5, we enforces that each known relation is a subrelation of the partial order $O$, meaning that $O$ must respect all the orders that are already known from the event-driven semantics.
Then, for each handler, and for every pair of messages $m_{1}, m_{2}$ on this handler, we require 
$(m_{1}.\text{last} <_{O} m_{2}.\text{first}~~~ \vee~~~ m_{2}.\text{last} <_{O} m_{1}.\text{first})$
where $m.\text{first}$ and $m.\text{last}$ denote the first (respectively last) instruction of a message $m$. This constraint enforces that the event order ($\eo$) is total within each handler, i.e., the execution of messages is serial. Similarly, we let $m.\text{post}$ denote the posting event of message $m$, and require
$(m_{1}.\text{post} <_{O} m_{2}.\text{post}~~~\vee~~~m_{2}.\text{post} <_{O} m_{1}.\text{post})$
which corresponds to the requirement of $\mo$ being total for the posts to a given handler.
Finally, the Z3 solver is called to check if the given set of constraints are satisfiable. 
If the solver returns Yes, then there exists a global partial order $O$ that extends the known relations and satisfies all handler-level ordering constraints and therefore, the trace is consistent.
Otherwise, no such order exists, which implies that the trace is inconsistent.

%% file: experiments.tex
\section{Implementation and Experimental Evaluation}
\label{sec:exp}

We have implemented a prototype tool for consistency checking of event-driven traces, based on the algorithms described in Section~\ref{sec:proc}. The prototype verifies whether a given partial trace admits a consistent extension, and, when successful, produces a witness: a concrete assignment to the missing relations (i.e., $mo$ and $eo$). From this witness, a valid execution can be reconstructed. All experiments were conducted on a machine running Debian 12.4 with an Intel(R) Xeon(R) Platinum 8168 CPU @ 2.70GHz and 192 GB of RAM.
Selected results are presented in Table~\ref{tab:experiment_results_droidracer}. Further results can be found in Appendix~\ref{app:experiments}.
\vspace{0.1in}

\input{table_exp-droidracer}

\smallskip

\noindent{\bf Experiment setting.}
To evaluate our approach, we used two independent methods to generate event-driven traces:

\smallskip

\noindent{\em Synthetic Traces via \Nidhugg.}
Our first method involves using the open-source model-checking tool \Nidhugg\
to generate traces  from ~\cite{Kragl20} and new synthetic programs\footnote{ A detailed discussion of these benchmarks is given in Appendix~\ref{app:experiments}.}. While \Nidhugg supports an event-driven execution model, it interprets asynchronous semantics using multisets rather than FIFO queues. As a result, some of the generated traces may be inconsistent under queue semantics. For each benchmark program, we randomly sampled traces from \Nidhugg's output. These traces are guaranteed to satisfy multiset semantics, but may violate the stricter queue-based consistency.

\smallskip

\noindent{\em Android Traces via \Droidracer.}
Our second source of benchmarks is \Droidracer~\cite{Maiya:pldi14}, a tool for systematic exploration of Android application behaviors. \Droidracer 's Trace Generator executes Android binaries on an emulator and exhaustively generates event sequences up to a bound k using depth-first search. We developed a custom parser to transform these sequences into partial traces suitable for our tool. These traces extracted from Android apps available at~\cite{Maiya2014DroidRacerArtefact}, originally used by Maiya et al.~\cite{Maiya:pldi14}. Static edges (e.g., program order, reads-from, and posted-by) are added during parsing, while the other relations are left unspecified. The tool then checks whether a consistent extension exists. Since Android's semantics closely follow queue-based message handling, we expect all \Droidracer traces to be consistent.

\smallskip

\noindent{\bf {Experimental Results}}
We compared the performance of our two algorithms described in Section~\ref{sec:proc} for event-driven consistency.

The results, shown in Table~\ref{tab:experiment_results_droidracer},
clearly demonstrate the advantage of Algorithm 2 in both perfromance and scalability. While Algorithm 1 performs acceptably on small traces (e.g Aardict), it fails to scale to complex instances due to the combinatorial explosion in possible execution orderings. In contrast, Algorithm 2 benefits from Z3's efficient constraint-solving capabilities, enabling fast detection of consistency or inconsistency even in challenging benchmarks.

To summarise, our experiments indicate that SMT-based techniques can be effectively leveraged for consistency checking in event-driven programs. The integration of SMT solvers, which are already widely adopted in verification tools, provides a scalable and precise foundation for reasoning about partial traces.

%% file: table_exp-droidracer.tex

\begin{table}[!tbp]
	\vspace{-\baselineskip}
	\centering
	\resizebox{\textwidth}{!}{
		\vspace{-\baselineskip}
	\begin{tabular}{| l | c | c | c | c | c | c | c | c | c | c |}
\hline
\multirow{3}{*}{\textbf{Benchmark}} &
&
&
&
\multirow{3}{*}{\# T} &
\multicolumn{3}{c|}{Algorithm 1} & \multicolumn{3}{c|}{Algorithm 2}\\
& Max & Max & Max & & \# Consistent & \# T/O & Time & \# Consistent & \# T/O & Time\\
& \# E & \# M & \# H & & traces & traces & in sec. & traces & traces & in sec.\\
	  \hline
SampleApp & 4776 & 13 & 5 & 1 & 0 & 1 & - & 1 & 0 & 2.9816 \\
Tomdroid & 4776 & 13 & 5 & 2 & 0 & 2 & - & 2 & 0 & 1.8621 \\
Opensudoku & 11292 & 14 & 5 & 1 & 0 & 1 & - & 1 & 0 & 22.0052 \\
Sgtpuzzles & 18406 & 18 & 5 & 2 & 0 & 2 & - & 2 & 0 & 34.6982 \\
Remindme & 6870 & 23 & 5 & 1 & 0 & 1 & - & 1 & 0 & 16.7987 \\
Modelcheckingserver & 4057 & 26 & 4 & 1 & 0 & 1 & - & 1 & 0 & 3.6556 \\
Messenger & 5034 & 26 & 7 & 2 & 0 & 2 & - & 2 & 0 & 5.4606 \\
Music & 4253 & 33 & 5 & 3 & 0 & 3 & - & 3 & 0 & 4.3871 \\
Fbreader & 6159 & 35 & 8 & 6 & 0 & 6 & - & 6 & 0 & 9.4964 \\
K9Mail & 5309 & 46 & 9 & 2 & 0 & 2 & - & 2 & 0 & 13.2162 \\
Aarddict & 1220 & 12 & 5 & 2 & 1 & 1 & 2.0932 & 2 & 0 & 0.0982 \\
AdobeReader & 15717 & 140 & 7 & 1 & 0 & 1 & - & 0 & 1 & - \\
Facebook & 5319 & 14 & 12 & 1 & 0 & 1 & - & 1 & 0 & 1.9775 \\
Twitter & 9889 & 30 & 12 & 1 & 0 & 1 & - & 1 & 0 & 18.7847 \\
Browser & 9762 & 34 & 15 & 6 & 0 & 6 & - & 3 & 3 & 46.9663 \\
Flipkart & 116945 & 61 & 15 & 1 & 0 & 1 & - & 0 & 1 & - \\
Mytracks & 3671 & 32 & 16 & 3 & 1 & 2 & 27.7057 & 3 & 0 & 2.2841 \\

	\hline
	\end{tabular}
	}
	\caption{
		\footnotesize{Experimental results for benchmark programs collected from droidracer. The field \# T denotes the number of traces. The traces can differ in size (events \# E), messages \# M, handlers \# H), and the field contains the maximum of its traces. The field \# Consistent traces denotes the number of these traces for which the implementation reports the existence of a satisfying execution. The field \# T/O traces denotes the number of traces for which our tool timed out (with a timeout of 120s). For any remaining traces, the tool concludes inconsistency. The time fields represent the average runtime for the traces that did not time out. A value of - indicates that the corresponding algorithm timed out on every trace.}
}
	\label{tab:experiment_results_droidracer}
	\vspace{-\baselineskip}
	\end{table}

%% file: conclusion.tex
\section{Conclusion and Future Work}

In this paper, we investigate the problem of ED consistency under the sequential consistency memory model.
We propose axiomatic semantics for event-driven programs and show equivalence of axiomatic and operational semantics. Furthermore, we establish that checking event-driven consistency is NP-hard,
even when the number of handler threads is bounded.
Further, when there is no nested posting in the trace, we show that checking consistency can be done in polynomial time.  
Finally, we also implement our event-driven consistency checking in a prototype tool, and provide promising experimental results on standard event-driven examples.

In the future, we plan to extend this work to the setting of other memory models such as Release-Acquire, Total Store Ordering etc.
We also plan to integrate our implementation in procedures for Dynamic partial order reduction for event-driven programs, race detection and predictive analysis.
Finally, we also plan to use evaluate our implementation on a wider range of examples, for instance the traces generated from android applications.

%% file: app-model.tex
\section{Appendix for Section~\ref{sec:model}}
\label{sec:app-model}

A transition occurs on either the execution of an instruction or a $\get$ which corresponds to receiving a message. The rules dictating these transitions are shown in \cref{app:figtransRules}. We write $\alpha \xrightarrow{a} \alpha'$ to denote a transition. 

\begin{figure}[!htbp]
\flushleft\text{\large\textbf{Event-transitions}}
\flushleft\textsc{Write}\\
\AxiomC{$\inst(\alpha.s_h.\line)=l_i \colon x=a \wedge \alpha.s_h.\val(a)=v$}
\UnaryInfC{$\alpha \xrightarrow{\tuple{h,write,x,v}} \alpha(\nu \leftarrow \nu(x\leftarrow v), s_h.\line \leftarrow \succ(s_h.\line))$ }
\DisplayProof
\flushleft\textsc{Post}\\
\AxiomC{$\alpha.s_h.\line=l\colon post(h',m)$\;\;$\alpha.s_{h'}.\beta \xrightarrow{\post,(m,newmid)} \beta'$\;\;$newmid=(h,\alpha.s_h.mcount)$}
\UnaryInfC{$\alpha \xrightarrow{\tuple{h,\post,h',newmid}} \alpha(s_{h'}.\beta \leftarrow \beta',s_h.mcount \leftarrow s_h.mcount +1, s_h.\line \leftarrow \succ(s_h.\line))$ }
\DisplayProof
 \flushleft\textsc{Read}\\
  \AxiomC{$\inst(\alpha.s_h.\line)=l\colon a=x$}
  \UnaryInfC{$\alpha \xrightarrow{\tuple{h,read,x}} \alpha(s_h.\val \leftarrow s_h.\val(a\leftarrow \nu(x)), s_h.\line \leftarrow \succ(s_h.\line))$ }
  \DisplayProof
\flushleft\textsc{Get}\\
\AxiomC{$\alpha.s_h.\line=l \colon last$ \;\;$\alpha.s_h.\beta_h \xrightarrow{\get,(m,mid)}_\adt \beta'$}
\UnaryInfC{$\alpha \xrightarrow{\tuple{h,\get,mid}} \alpha(s_h.\beta \leftarrow \beta', s_h.mid \leftarrow mid,s_h.\line \leftarrow m.\init)$ }
\DisplayProof

\text{\large\textbf{Local-transitions}}
\flushleft\textsc{IntWrite}\\
\AxiomC{$\inst(\alpha.s_h.\line)=l\colon a=exp$}
\UnaryInfC{$\alpha \xrightarrow{\varepsilon} \alpha(s_h.\val(a)\leftarrow \val(exp), s_h.\line \leftarrow \succ(s_h.\line))$}
\DisplayProof
\newline
\flushleft\textsc{IfCond}\\
\AxiomC{$\inst(\alpha.s_h.\line)=l\colon if \ cond \ \textbf{goto } l'$}
\AxiomC{$\alpha.s_h.\val(cond) = \true$}
\BinaryInfC{$\alpha \xrightarrow{\varepsilon} \alpha(s_h.\line \leftarrow l')$ }
\DisplayProof
\newline
\flushleft\textsc{IfNotCond}\\
\AxiomC{$\inst(\alpha.s_h.\line)=l\colon if \ cond \ \textbf{goto } l'$}
\AxiomC{$\alpha.s_h.\val(cond) = \false$}
\BinaryInfC{$\alpha \xrightarrow{\varepsilon} \alpha(s_h.\line \leftarrow \succ(l))$ }
\DisplayProof
\flushleft\textsc{Goto}\\
\AxiomC{$\inst(\alpha.s_h.\line)=l\colon \textbf{goto } l'$}
\UnaryInfC{$\alpha \xrightarrow{\varepsilon} \alpha(s_h.\line \leftarrow l')$ }
\DisplayProof
  \caption{Transition rules of programs}
  \label{app:figtransRules}
\end{figure}

%% file: app-semantics-eq.tex
\section{Equivalence of Operational and Axiomatic Semantics}
\label{app:sec:axSemanticsEqProof}

    \begin{wrapfigure}[20]{r}{6cm}
     \centering
        \centering
        \includegraphics[width=0.5\textwidth]{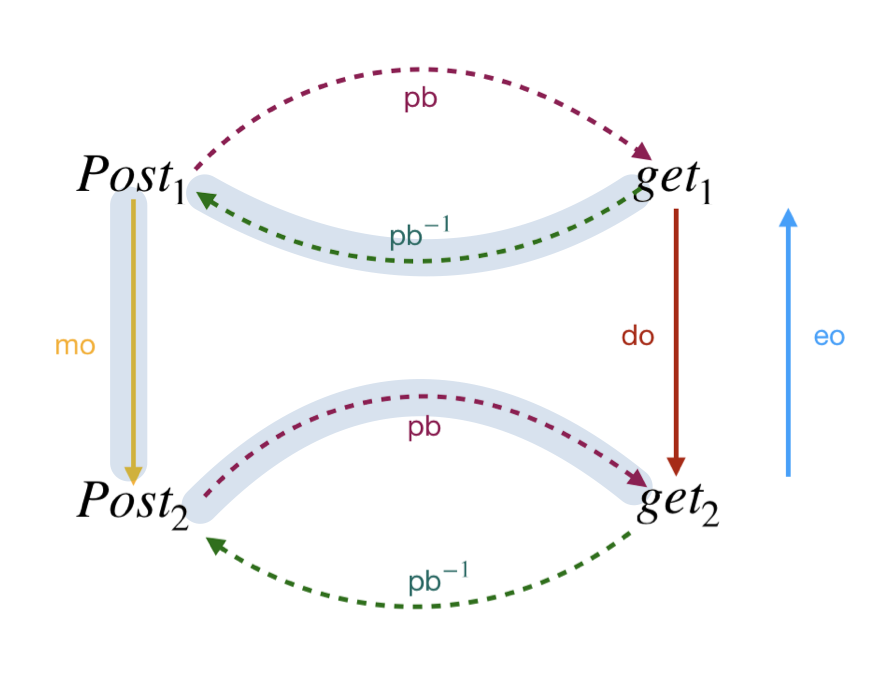}
        \caption{Depiction of $\dto$ edge for the mailbox of a handler
         Here, the $\eo$ edges that violates the queue semantics is depicted in blue, and the violating cycle caused by this violation may be observed.}
        \label{fig:do-queue}
    \end{wrapfigure}
We recall that we work with the set of relation $\rels = \{ \rf, \co, \po, \eo, \pb, \mo\}$. 
A trace is a directed graph $\tau = (E, \Delta)$ where $E$ is a finite set of events, $\Delta \subseteq E \times \rels \times E$ is a set of edges on $E$ with labels from $\rels$. 
A \emph{partial trace} $\tau'=(E', \Delta')$ is said to \emph{extend} a trace $\tau =(E, \Delta)$ if $E \subseteq E', \Delta \subseteq \Delta'$. 
A \emph{linearisation} $\pi=(E,\leq_\pi)$ of a trace $\tau=(E,\Delta)$ is a total ordering $\leq_\pi$ satisfying $\delta \in \Delta \Rightarrow \delta \in \leq_\pi$.

\textbf{Remark.}    
Note that we also sometimes have to deal with traces that are not well-formed, i.e., where not every $\post$ event has a corresponding $\get$ event.
We will need these notions in the proofs below. 
However, the traces we consider as input for consistency checking will always be well-formed.

Further, given a program $\mP$ and its execution $\rho$, recall that the event set $E_{\rho}$ along with the total order $\leq_\rho$ derived from the run induces a trace $\tau(\rho)$.

\begin{theorem}~\label{thm:Eq-axiomatic-operational}
    A trace $\tau$ is axiomatically consistent iff there exists an event-driven program $\mP$ and an execution $\rho$ of $\mP$ such that $\tau=\tau(\rho)$. 
\end{theorem}
  
\begin{proof}
    For the forward direction, suppose that a trace $\tau$ is axiomatically consistent.     
    We need to show that there exists a program $\mP$ such that it has a run $\rho$ that satisfies $\tau(\rho)=\tau$.

    Let $\tau = (E, \Delta)$, where $E$ is the set of events and $\Delta$ is the set of edges. Since the trace $\tau$ is axiomatically consistent, we know that $\Delta$ is acyclic.
    Consider a linearisation $\sigma$ of the trace $\tau$ of the form
    $e_1 \cdot e_2 \cdot e_3 \cdots e_n$.
    Let $\sigma_i$ denote the prefix of $\sigma$ containing the first $i$ events and let $\tau_i$ denote the projection of $\tau$ to the first $i$ events of $\sigma$ i.e., $\tau_i = \tau\downarrow_{E_{i}}$.

    We will produce a program $\mP$ and a run $\rho$ of $\mP$ such that $\tau(\rho) = \tau$, where
    $$\rho:= C_0 \xrightarrow{\bar{e}_1} C_1 \xrightarrow{\bar{e}_2} C_2 \xrightarrow{\bar{e}_3} C_3 \cdots $$
    where 
    $\bar{e}_i$'s are event transitions.
    To show that $\tau(\rho) = \tau$, we need to show that 
    \begin{itemize}
        \item The set of events of $\tau(\rho)$ is precisely the set $E$ of events of $\tau$. For this, it suffices to show that $\bar{e}_i = e_i$ for all $1 \leq i \leq n$.
        \item The set of relations in $\tau(\rho)$ are in agreement with the the relations in $\tau$, i.e.,
        $$e_1~ R~ e_2 \iff \bar{e}_1~  R~  \bar{e}_2$$ 
        where $R$ is a relation in $\rels$.
    \end{itemize}

    Next, we will produce the program $\mP$ by defining its building blocks as follows.
    \begin{itemize}
        \item The set of handlers $H$ is given by the set 
        $\{h \mid \tuple{h,post,h'} \text{ is an event in $E$}\}$.
        In other words, the set of handlers are determined by the set of post events in $\tau$, as each message uniquely identifies its handler. 
        \item The set of variables are given by the set of variables $x$ such that either $\tuple{h,write,x,v}$ or $\tuple{h,read,x}$ is an event in $\tau$.
        \item The set of registers is given by $\{a_h \mid h \in H\}$, i.e., 
        each handler has a register. 
        Note that we need just one register per handler, as all the reads can be read to this register.
        \item The set of messages of the handler is given the set of $\get$ events of a handler. 
        For a handler $h$, consider the projection of $\tau$ to the events of $h$. Then, the set of events reachable by a sequence of $\po$ edges from a $\get$ event constitute the instructions of the message $m$
        \begin{itemize}
            \item For a read event of the form $\tuple{h,read,x}$, we add the instruction $a_h = x$ to the current message of handler $h$.
            \item For a write event of the form $\tuple{h,write,x,v}$, we add the instruction $x = v$ to the current message of handler $h$.
            \item For a post event of the form $\tuple{h,post,h'}$, 
            we instantiate a post instruction of the message $\post(h,m)$ in the current active message of the handler $h$.
        \end{itemize}
        Eventually, this $\po$ path will lead to an $last$ event, which indicates the end of the message, at which point, we stop.
        \item Note that an exception to this is the initial messages in each handler, which do not start with a $\get$ event. 
        Consequently, any trace $\tau$ has $n$ minimal events which do not have any incoming edges, one for each handler. 
        To address this, for each handler, we construct the initial message of the handler with the minimal event in the trace $\tau$ corresponding to that handler, and keep adding the set of events reachable by a sequence of $\po$ edges from this event.
        As mentioned above, the end of these messages will be indicated by an $last$ event.
    \end{itemize}
        
    Now that we have constructed the program $\mP$, we will show that it has a run $\rho$ that satisfies $\tau(\rho)=\tau$,
    i.e., the event-driven program relations in $\tau(\rho)$ agree with the respective relations imposed by $\tau$. 
    We will show this inductively.
    Recall that $\tau_i$ denotes the projection of $\tau$ to the first $i$ events of $\sigma$. 
    We will denote by $\rho_i$ the prefix of $\rho$ containing the first $i$ events. 
    Then, we will show that 
    $$\tau(\rho_i) = \tau_i$$
    Note that $\tau_i$ is not a well-formed trace.
    The proof follows by induction on the number of events in $\tau$. 

    For the base case, consider the empty trace $\tau_0$.
    This corresponds to the initial configuration $C_0$ of the program, where the mailboxes are empty and all the variables and registers have the initial value, and all the mailboxes are empty. 

    Let $C_i$ be the configuration reached by the execution $\rho_i$, i.e., by executing the sequence of events $e_1 \cdot e_2 \cdot e_3 \cdots e_i$.
    Note that in $C_i$, all the shared variables and registers will have the value written by the latest write instruction involving that variable (respectively register).
    Further, let $C_i \xrightarrow{e_{i+1}} C_{i+1}$. 
    We will show that $\tau(\rho_{i+1}) = \tau_{i+1}$.

    We will consider four cases, depending on the type of the event $e_{i+1}$.
    \begin{itemize}
        \item \textbf{Read event.} Suppose $e_{i+1} = \tuple{h,read,x}$. 
        Due to sequential consistency, the value read by $a_h$ will be the value written by the last write event on $x$ in the sequence $\rho_i$, say $e_k$ where $k < i$, introducing an $\rf$ edge from $e_k$ to $e_{i+1}$ in $\tau(\rho_{i+1})$.   
        Since $\sigma$ is a linearisation of $\tau$, the $\rf$ edge to $e_{i+1}$ in $\tau$ should also agree with this.  

        \item \textbf{Write event.}
        Suppose $e_{i+1} = \tuple{h,write,x,v}$. 
        This will introduce $\co$ edges from all the write events on $x$ in $\rho_i$ to $e_{i+1}$  in $\tau(\rho_{i+1})$. 
        Once again, since $\sigma$ is a linearisation of $\tau$, these $\co$ edges are consistent with the $\co$ edges in $\tau_{i+1}$.  

        Additionally, these two events also induce $\po$ edges from the preceding event $e_{i}$ of $\rho_i$, whose consistency also follows by the same argument. 
        
        \item \textbf{Post event.}
        Suppose $e_{i+1} = \tuple{h,post,h'}$. 
        This will introduce $\mo$ edges to $e_{i+1}$ from all the post events in $\rho_i$ that have posted to the same handler. 
        The consistency of these edges follows from the fact that $\sigma$ is a linearization of $\tau$.

        \item \textbf{Get event.}
        Suppose $e_{i+1} = \tuple{h,\get}$.
        This will introduce 
        \begin{itemize}
            \item a $\pb$ edge from the post event $\tuple{h',post,h}$ corresponding to this message. 
            \item $\eo$ edges from the get events $\tuple{h,\get}$ seen before in the sequence $\rho_i$.
        \end{itemize}
    
    \end{itemize}
    Thus, for all the possibilities of $e_{i+1}$, we have shown that $\tau(\rho_{i+1}) = \tau_{i+1}$.

    \vspace*{0.2in}

    For the reverse direction, suppose that a program $\mP$ has an execution $\rho$ that induces the trace $\tau(\rho)$. 
    We need to show that $\tau(\rho)$ is axiomatically consistent. 
    From the definition of the trace of a program, we know that $\tau(\rho)$ is 
    of the form $\tau = (E, \leq)$, where 
    \begin{itemize}
        \item $E = E_{\rho}$ is the set of events in $\rho$ and 
        \item $\leq$ is a total order on $E_{\rho}$ defined as $e_i \leq_{\rho} e_j$ iff $i \leq j$, where $e_i, e_j \in E$. 
    \end{itemize}
    From Definition~\ref{defn:axCons}, we know that $\tau$ is said to be axiomatically $\adt$-consistent if the relation $(\po \; \cup \; \rf  \; \cup \; \fr  \; \cup \; \co  \; \cup \; \pb  \; \cup \; \mo  \; \cup \; \eodag  \; \cup \; \dto )$ is acyclic.

    It is easy to see that no cycles are created by $\po$, $\rf$, $\co$, $\pb$, $\mo$ and $\eo$ edges, as this follows from the equivalence of axiomatic and operational models for sequentially consistent programs. 
    Further, since $\fr$ is derived from $\rf$ and $\co$ edges, addition of $\fr$ edges do not create cycles. 
    
    It remains to show that adding the $\eodag$ and $\dto$ edges also do not create cycles. 

    \begin{itemize}
        \item $\eodag$ edges: 
        Note that $\eodag$ edges only exists between events of a message $m$ and a message $m'$ of the same handler, such that there is an $\eo$ edge between the $\get$ event of $m$ and $\get$ event of $m'$. 
        Thus, they only add edges between events of a handler. 
        Since all the events of a handler are totally ordered, and the $\eodag$ respects this total order (as it is inherited from $\eo$), these edges respect the order given by $\rho$. 
        
        \item $\dto$ edges: 
        Before discussing the $\dto$ edges that are induced, we recall that the order of insertion and deletion of messages into the mailbox are governed by the $\mo$ and $\eo$ edges. 
        Recall that the $\mo$ relation orders all the post events involving  messages posted to the same handler, and the $\eo$ relation orders all the get events of the messages of a handler.
            
            Further, recall that $\dto=\pb^{-1}.\mo.\pb$.
            Since the $\mo$ relation orders events that post messages to the mailbox of a handler, 
            it is easy to see that the new $\dto$ edges that are induced  between $\get$ events of the messages posted to the mailbox of the same handler. 
            It suffices to show that the new edges introduced agree with the $\eo$ edges between these $\get$ events. 

            Note that the operations of each handler in the run $\rho$ follows the queue semantics, i.e., the $\get$ instructions should be processed in first-in-first-out manner.  
            Further, the relation $\dto=\pb^{-1}.\mo.\pb$ orders the $\get$ events of the messages posted to the mailbox of any given handler in the order in which $\post$ events are ordered in $\tau$. 
            Therefore, if any two $\get$ events, say $\get_1$ and $\get_2$, violate the queue semantics, then the $\eo$ edge between these events will cause a cycle - there will be a cycle of the form $\get_1 ~\dto~ \get_2 ~\eo~ \get_1$, , as depicted in Figure~\ref{fig:do-queue}.             

            Formally, suppose there is a cycle involving $\dto$ edges. Then, there is a $\dto$ cycle involving two $\get$ events, say $\get_1$ and $\get_2$. 
            The cycle is of the form $\get_1 ~\dto~ \get_2 ~\hb~ \get_1$.
            However, 
            the existence of the $\dto$ edge from $\get_1$ to $\get_2$ implies that the message corresponding to $\get_1$ was added to the queue before the message corresponding to $\get_2$. 
            Further, the $\hb$ edge from $\get_2$ to $\get_2$ implies that $\get_2$ is processed before $\get_1$, which means that the execution of messages corresponding to $\get_1$ and $\get_2$ violate the queue semantics. 
            This is a contradiction to the assumption that $\sigma$ was a consistent execution.  
    \end{itemize}
\end{proof}

%% file: app-np-hardness-bddhandler-queue.tex
\section{Formal Proof of NP-Completeness of ED-Consistency}
\label{app:hardness}

In this section, we provide a formal proof of Theorem~\ref{thm:NP-hardness-bddhandler-queue}, which states that the ED-consistency problem is NP-complete even when the number of handlers is bounded.

\subsection*{Formal Construction}
The proof is done by reduction from 3-BI-3SAT.
Let $\phi$ be a 3-BI-3SAT instance with variables $x_1,x_2,\ldots,x_n$ and clauses $C_1,C_2, \ldots, C_m$. We will construct a partial ED trace $\tau=(E,\Delta)$, with $\Delta \subseteq E \times \rels' \times E$, such that $\tau$ can be extended to a axiomatically consistent  trace  $\tau'=(E,\Delta')$, with $\Delta' \cap (E \times \rels' \times E)=\Delta \cap (E \times \rels' \times E)$ iff $\phi$ is satisfiable.
\vspace{0.1in}

\noindent\textbf{High level structure. }
The construction of the trace $\tau$ is divided into two stages, which we call Stage 1 and Stage 2 respectively. There are 8 handlers in Stage 1 and 5 handlers in Stage 2. One handler, namely $h_W$ is common to both stages, hence totally there are 12 handlers. If a satisfying assignment exists for  $\phi$, then there is a program which can execute the events in Stage 1 followed by those in Stage 2, i.e., $\tau$ is consistent. If $\phi$ is unsatisfiable then there is no witnessing execution possible which executes both stages and $\tau$ is inconsistent. 

\underline{Stage 1} corresponds to the selection of a satisfying assignment $f$ for $\phi$.  
 We can encode the information of whether a variable $x_i$ is assigned true or false using the order of execution of two messages $m_{i,1}$ and $m_{i,0}$ on the same handler, where $x_i$ is assigned true (resp. false) if $m_{i,1}$ (resp. $m_{i,0}$) is executed later. Unfortunately, this will not work due to technical difficulties faced in clause verification which we explain when describing Stage 2.
  This necessitates our extremely technical reduction which makes use of the structure of the 3-BI-3SAT instance where each variable occurs in \emph{at most} three clauses and the variables occurring in a clause are all different. We have to create (at most) 3 copies of the messages, one for each clause in which $x_i$ occurs and find a way to synchronise the assignment between these three copies. 

Hence the messages for $x_i$ are actually of the form $m_{i,j,b}$ where $j$ refers to clause $C_j$ and $b \in {0,1}$. Using the technique of \emph{nested postings} (explained later in full construction), we post the set $M$ of $m_{i,j,b}$  messages in the queue of $h_W$ in some order $\sigma$. The remaining 7 handlers of Stage 1 are used to shuffle the messages with certain restrictions on the order $\sigma$ of messages. 

 The set $S$ of all the possible orders $\sigma$ is such that, every $\sigma$ is constrained to be \emph{consistent} (we explain later in full construction) with some particular assignment $f$ of variables of $\phi$. There are no other constraints on the order $\sigma$. At the end of Stage 1, the queues of all other Stage 1 handlers is empty and the queue of $h_W$ is populated in some order $\sigma \in S$ consistent with some assignment $f$. Note that there are multiple $\sigma$ which are consistent with a particular $f$, this fact will be important later.

\underline{Stage 2} Let us fix $\sigma$ and $f$ from Stage 1. Then Stage 2 verifies that $f$ indeed satisfies all the clauses of $\phi$. For this, we build a clause gadget $G_j$ corresponding to each clause $C_j$. The set $E_G$ of events of these clause gadgets occupy the 4 non-$h_W$ handlers of Stage 2. The $E_G$ events belong to an initial message of each of the 4 handlers, and consist purely of read and write events. Recall that the queue of $h_W$ is populated at this point with messages $M$. The information regarding the assignment is encoded in the order of the messages in $h_W$. This information is transferred to the other 4 non-$h_W$ handlers via a technique we call \emph{sandwiching} (explained in the full construction). There are now two possibilities:\\ 
 \noindent(1) If $f$ is not a satisfying assignment, then some clause $C_j$ is not satisfied by $f$. In this case, any order $\sigma$ of messages consistent with $f$ will induce a $\hb$ (happens before) cycle in the corresponding gadget $G_j$ via the sandwiching. Therefore Stage 2 cannot be executed by any witnessing execution. If there are no satisfying assignments, then $\phi$ is unsatisfiable and hence $\tau$ is not consistent.  

 \noindent(2) If $f$ is a satisfying assignment then there is some order $\sigma$ of the messages in $M$ which is consistent with $f$ such that there is a witnessing execution. The clause gadgets are executed interleaved with the messages in $M$ due to the sandwiching. The execution happens sequentially i.e. $G_1$ is executed, then $G_2$ etc. This implies $\tau$ is consistent. 

\noindent\textbf{Full Construction}

For now we assume that at the end of Stage 1, all of the messages in  $M$ have been posted to $h_W$ in some order $\sigma$ consistent with some assignment $f$ to the variables of $\phi$. We describe how we can check the satisfaction of clauses in Stage 2 before describing Stage 1. 

\vspace{0.1in}

\noindent\textbf{Stage 2: Clause checking.}
Stage 2 events occur in the 5 handlers $h_{C_a}, h_{C_b}, h_{C_c}, h_{C_d}, h_W$.
For each clause $C_j$ we create a clause gadget $G_j$ which consists of 14 events in handlers $h_{C_a}, h_{C_b}, h_{C_c}, h_{C_d}$. Pick a clause, say $C_2=x_1 \vee x_2 \vee \overline{x_n}$.  
Figure \ref{app:fig:clauseGadget} shows the clause gadget $G_2$ for clause $C_2$  together with two of the messages $m_{i,j,b}$ in $h_W$ which are posted by Stage 1. For reasons of space we use $W$ and $R$ for $write$ and $read$ in the description of events. We will use program variables of two forms: (F1) $l^{j}_{i,k}$ and $\overline{l^{j}_{i,k}}$, and (F2) $z_k$. The F2 variables do not correspond in any way to the formula, while the F1 variables do.

\underline{Clause satisfaction.}  First let us focus on the dotted boxes $b_1$ and $b_2$ in the figure. Each box consists of a read followed in $\po$ by a write event. Focusing on $b_1$, the program variable $\overline{l^2_{n,1}}$ in $e_{10}$ has the information: superscript $2$ for clause $C_2$, first subscript $n$ indicating the literal $\overline{x_n}$ and second subscript $1$ indicating it is the first event in the box. Note that since we have made (up to) 3 different copies of each variable in our reduction of regular 3SAT to 3-BI-3SAT, the variables $l^{j}_{i,k}$ and $\overline{l^{j}_{i,k}}$ occur exactly once. The events in $b_1$ are linked to the read and write events in the message $m_{n,2,0}$ via $\rf$ arrows. The direction of the arrows implies that the events in box $b_1$ have to be executed after event $e_{15}$ and before $e_{16}$ which are both in message  $m_{n,2,0}$. This is the technique we call \emph{sandwiching}. Similarly $b_2$ has to be executed during the execution of $m_{n,2,1}$. Suppose $m_{n,2,0}$ is executed before $m_{n,2,1}$ as indicated by the $\eo$, this means that $x_n$ is assigned the value $\true$. The sandwiching induces the red $\hb$ relation shown between $e_{11}$ and $e_{12}$, \emph{copying} the value of the variable from handler $h_W$ to the clause gadget by ensuring that the read and write events of the $\overline{l^2_{n,1}}$ and $\overline{l^2_{n,2}}$ program variables occur before the events on the $l^2_{n,1}$ and $l^2_{n,2}$ variables, which again reflects that the $x_n$ has been set to true in $C_2$.

 Notice that similar boxes can be drawn around events $e_2,e_3$ and $e_4,e_5$ corresponding to copying the assignment to variable $x_1$ and for $e_6,e_7$ and $e_8,e_9$ for variable $x_2$. Further note that the F1 program variables occur in a certain order when we traverse $h_{C_a}, h_{C_b}, h_{C_c}, h_{C_d}$ from top to down and from $a$ to $d$. Let us skip the second subscript (which is just used to denote two copies) and see the order in this example: $l^2_1$, $\overline{l^2_1}$, $l^2_2$, $\overline{l^2_2}$, $\overline{l^2_n}$, $l^2_n$. The fact that $l^2_1$ occurs before $\overline{l^2_1}$ indicates that $x_1$ is present in positive form (as also $l^2_2$ before $\overline{l^2_2}$ representing the occurrence of $x_2$ in positive form). Whereas $\overline{l^2_n}$ before $l^2_n$ program variables indicates that $\overline{x_n}$ is present in $C_2$. In this way, the clause gadget captures the structure of the clause.

 Each of the other boxes around events $e_2,e_3$ and $e_4,e_5$ has similar sandwiching $\rf$ relations to messages $m_{i,j,b}$ in $h_W$ which are not shown in the figure. The three red $\hb$ arrows correspond to setting each of the three variables in $C_2$ to a value that falsifies the corresponding literal in $C_2$. The events $e_1$ and $e_{14}$ use a variable $z_2$ (where the subscript refers to the clause $C_2$) and are connected by an $\rf$. Thus if $m_{1,2,1} \; \eo \; m_{1,2,0}$, $m_{2,2,1} \; \eo \; m_{2,2,0}$ and $m_{n,2,0} \; \eo \; m_{n,2,1}$ all hold, a cycle is formed and the clause gadget cannot be executed. On the other hand, if even one of the red  arrows is flipped (indicating that a literal of $C_2$ is set to $\true$), then the arrows form a partial order allowing execution of the clause gadget $G_2$. Lastly, note that a variable e.g. $x_2$ in $C_2$ also occurs in $C_1$ and $C_3$. In Stage 1 we explain how we can select an assignment for $x_2$ in a consistent way for all three clauses $C_1,C_2,C_3$. 
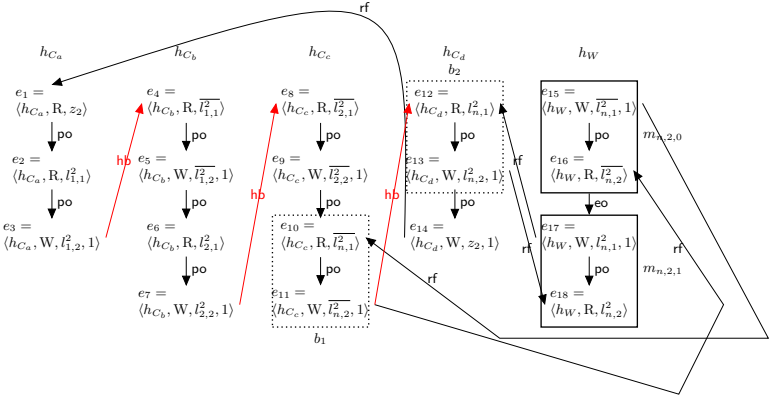
\begin{figure}[h]
        \vspace{-\baselineskip}
        \centering
        \resizebox{0.85\textwidth}{!}{
\begin{tikzpicture}
\foreach \i in {1,...,5} {
    \coordinate (H\i) at ({(\i-1) * \nodehGap},{-0.25* \nodevGap}); 
}

\node (hCa) [minimum width=\nodeWidth, minimum height=\nodeLen] at (H1) {$h_{C_a}$}; 

\node (hCb) [minimum width=\nodeWidth, minimum height=\nodeLen] at (H2) {$h_{C_b}$}; 

\node (hCc) [minimum width=\nodeWidth, minimum height=\nodeLen] at (H3) {$h_{C_c}$};

\node (hCd) [minimum width=\nodeWidth, minimum height=\nodeLen] at (H4) {$h_{C_d}$};

\node (hW) [minimum width=\nodeWidth, minimum height=\nodeLen] at (H5) {$h_W$};
\foreach \i in {1,...,3} {
       \coordinate (E\i) at (0,{-(\i ) * \nodevGap}); 
}
\foreach \i in {4,...,7} {
       \coordinate (E\i) at (\nodehGap,{-(\i - 3) * \nodevGap}); 
}
\foreach \i in {8,...,11} {
       \coordinate (E\i) at (2*\nodehGap,{-(\i - 7) * \nodevGap}); 
}
\foreach \i in {12,...,14} {
       \coordinate (E\i) at (3*\nodehGap,{-(\i - 11) * \nodevGap}); 
}

\foreach \i in {15,...,18} {
       \coordinate (E\i) at (4*\nodehGap,{-(\i - 14) * \nodevGap}); 
}

\ereadEvent{e_1}{E1}{h_{C_a}}{z_2};
\ereadEvent{e_2}{E2}{h_{C_a}}{l^2_{1,1}};
\ewriteEvent{e_3}{E3}{h_{C_a}}{l^2_{1,2}}{1};
\node (b1) [draw,dotted, line width=0.3mm, minimum width=4.3*\nodeWidth,minimum height=5*\nodeLen, label={south:$b_1$}] at ($(E10)!0.5!(E11)$) {};
\node (b2) [draw,dotted, line width=0.3mm, minimum width=4.3*\nodeWidth,minimum height=5*\nodeLen, label={north:$b_2$}] at ($(E12)!0.5!(E13)$) {};

\ereadEvent{e_4}{E4}{h_{C_b}}{\overline{l^2_{1,1}}};
\ewriteEvent{e_5}{E5}{h_{C_b}}{\overline{l^2_{1,2}}}{1};
\ereadEvent{e_6}{E6}{h_{C_b}}{l^2_{2,1}};
\ewriteEvent{e_7}{E7}{h_{C_b}}{l^2_{2,2}}{1};

\ereadEvent{e_8}{E8}{h_{C_c}}{\overline{l^2_{2,1}}};
\ewriteEvent{e_9}{E9}{h_{C_c}}{\overline{l^2_{2,2}}}{1};
\ereadEvent{e_{10}}{E10}{h_{C_c}}{\overline{l^2_{n,1}}};
\ewriteEvent{e_{11}}{E11}{h_{C_c}}{\overline{l^2_{n,2}}}{1};

\ereadEvent{e_{12}}{E12}{h_{C_d}}{l^2_{n,1}};
\ewriteEvent{e_{13}}{E13}{h_{C_d}}{l^2_{n,2}}{1};
\ewriteEvent{e_{14}}{E14}{h_{C_d}}{z_2}{1};

\ewriteEvent{e_{15}}{E15}{h_W}{\overline{l^2_{n,1}}}{1};
\ereadEvent{e_{16}}{E16}{h_W}{\overline{l^2_{n,2}}};
\ewriteEvent{e_{17}}{E17}{h_W}{l^2_{n,1}}{1};
\ereadEvent{e_{18}}{E18}{h_W}{l^2_{n,2}};

\node (m1) [draw, line width=0.3mm, minimum width=4.3*\nodeWidth,minimum height=5*\nodeLen, label={east:$m_{n,2,0}$}] at ($(E15)!0.5!(E16)$) {};
\node (m2) [draw, line width=0.3mm, minimum width=4.3*\nodeWidth,minimum height=5*\nodeLen, label={east:$m_{n,2,1}$}] at ($(E17)!0.5!(E18)$) {};


\vdrawPo{e_1}{e_2};
\vdrawPo{e_2}{e_3};

\vdrawPo{e_4}{e_5};
\vdrawPo{e_5}{e_6};
\vdrawPo{e_6}{e_7};

\vdrawPo{e_8}{e_9};
\vdrawPo{e_9}{e_{10}};
\vdrawPo{e_{10}}{e_{11}};

\vdrawPo{e_{12}}{e_{13}};
\vdrawPo{e_{13}}{e_{14}};

\vdrawPo{e_{15}}{e_{16}};
\vdrawPo{e_{17}}{e_{18}};

\draw[->] (e_{14}.west) .. controls (8,1.5) .. node[midway, above] {$\rf$} (e_1.north);

\draw[->] ($(e_{11}.east)$) -- (14,-8) -- (15,-6)-- (e_{16}.east) node[midway, below] {$\rf$};
\draw[->] (e_{15}.east) --(16,-6.75) --(10,-6.75) -- ($(e_{10}.east)$) node[midway, above] {$\rf$};

\draw[->] ($(e_{13}.east)$) -- (e_{18}.west) node[midway, below] {$\rf$};
\draw[->] (e_{17}.west) -- ($(e_{12}.east)$) node[midway, above] {$\rf$};
\draw[->] (m1.south) -- (m2.north) node[midway, right] {$\eo$};

\draw[->,color=red] (e_{11}.east) -- ($(e_{12}.west)$) node[midway, above] {$\hb$};
\draw[->,color=red] (e_7.east) -- ($(e_8.west)$) node[midway, above] {$\hb$};
\draw[->,color=red] (e_3.east) -- ($(e_4.west)$) node[midway, above] {$\hb$};
\end{tikzpicture}
}
\caption{The structure of clause gadget $G_2$ of $C_2$ for the example in Figure \ref{app:fig:varsAndClauses}.}
\label{app:fig:clauseGadget}
\vspace{-\baselineskip}
\end{figure}

The complete set of events in the handlers $h_{C_a}, h_{C_b}, h_{C_c}, h_{C_d}$ is made up of the union of the events in the clause gadgets $G_j$. 
The clause gadgets $G_1,G_1,\ldots, G_m$ are placed in that order in the handlers $h_{C_a}, h_{C_b}, h_{C_c}, h_{C_d}$ and connected by $\po$ arrows. For example, $G_j.e_3$ will be $\po$ before $G_{j+1}.e_1$ in $h_{C_a}$, $G_j.e_7$ will be $\po$ before $G_{j+1}.e_4$ in $h_{C_b}$, etc.  In other words, the events of each of these four handlers can be assumed to be in an initial message in the respective handlers. There is no posting of events either from or to these 4 handlers. 

\noindent\textbf{Stage 1: Variable assignment.}
In this Stage, we use the handlers $h_V,h_{t_1},h_{t_2},h_{t_3},h_{t_4}, h_{t_5},h_{t_6}$ in order to post messages to $h_W$. 
We stated that an assignment to variable $x_i$ in clause $C_j$ can be encoded as the order between two messages $m_{i,j,0},m_{i,j,1}$.\\ 
\underline{Challenge 1:} How can we ensure that two messages $m_{i,j,0},m_{i,j,1}$ can posted in any order to $h_W$? \\ 
In order to solve this, we use \emph{nested posting}. $h_V$ can post  messages $m'_{i,j,0}$ to $h_{t_1}$ and $m'_{i,j,1}$ to $h_{t_2}$. Then  $m'_{i,j,0}$ (resp. $m'_{i,j,1}$) posts $m_{i,j,0}$ (resp. $m_{i,j,1}$) to $h_W$. Since $m'_{i,j,0}$ and $m'_{i,j,1}$ are on different handlers, they can be executed in any order, thus ensuring that $m_{i,j,0},m_{i,j,1}$ can posted in any order to $h_W$. Next we take up the reason for using multiple messages for each variable. 

Consider the sandwiching technique that we presented in Stage 2 in order to copy the assignment of a variable to the clause gadget. Suppose we were to use a single pair of messages $m_{i,0},m_{i,1}$ in $h_W$ for a variable $x_i$ from which this value was copied to the different clauses in which $x_i$ occurs. This means that any handler in which a clause gadget is being executed would be blocked from running till all of the clauses containing $x_i$ are able to finish executing the boxes corresponding to $x_i$. This leads to a cascading set of blocked handlers, requiring an unbounded number of handlers to execute the clause gadgets. In order to overcome this difficulty, we have to use three copies of the two messages as mentioned before. But this leads to a different challenge:\\
\underline{Challenge 2:} How can we ensure that the different copies of the messages corresponding to the same variable encode the same value?\\ 
To address this, we further extend the nesting of posts. In order to understand how this is done, we have to look into the structure of the formula $\phi$. Figure \ref{app:fig:varsAndClauses} represents the occurrence of the $n$ variables in the $m$ clauses (along with which literal occurs by use of a bar). Note that each row contains 3 marked cells and each column contains 4 or 6 marked cells as per the restriction on 3-BI-3SAT. A \emph{post sequence} is a partial trace of the form
\begin{tikzpicture}
\foreach \i in {1,...,5} {
       \coordinate (E\i) at ({(\i - 1) * \nodeGap}, 0); 
}

\postEvent{e1}{E1}{h_0}{h_1};
\getEvent{e2}{E2}{h_1};

\postEvent{e3}{E3}{h_1}{h_2};
\node (e4) [minimum width=\nodeWidth, minimum height=\nodeLen] 
        at (E4) {$\ldots$};
\postEvent{e5}{E5}{h_{n-1}}{h_n};

\drawPb{e1}{e2};
\drawPo{e2}{e3};
\drawPb{e3}{e4};
\drawPo{e4}{e5};
\end{tikzpicture}
\begin{figure}
        \vspace{-\baselineskip}
                \centering
                \begin{tikzpicture}[scale=0.5]
        
                    \def\cellSize{0.5cm}
        
                    \foreach \i in {1,...,8} {
                        \foreach \j in {1,...,8} {
                            \draw[thick] (\j-1, -\i+1) rectangle (\j, -\i);
        
                            \coordinate (mid-\i-\j) at (\j-0.5, -\i+0.5);
                        }
                    }
        
                \node (x1) [minimum width=\nodeWidth, minimum height=\nodeLen] 
                        at ($(mid-1-1)-(1,0)$) {$x_1$};
                \node (barx1) [minimum width=\nodeWidth, minimum height=\nodeLen] 
                        at ($(mid-2-1)-(1,0)$) {$\overline{x_1}$};
                \node (x2) [minimum width=\nodeWidth, minimum height=\nodeLen] 
                        at ($(mid-3-1)-(1,0)$) {$x_2$};
                \node (barx2) [minimum width=\nodeWidth, minimum height=\nodeLen] 
                        at ($(mid-4-1)-(1,0)$) {$\overline{x_2}$};
                \node (ldots1) [minimum width=\nodeWidth, minimum height=\nodeLen] 
                        at ($(mid-5-1)-(1,0)$) {$\vdots$};
                \node (xn) [minimum width=\nodeWidth, minimum height=\nodeLen] 
                        at ($(mid-7-1)-(1,0)$) {$x_n$};
                \node (barxn) [minimum width=\nodeWidth, minimum height=\nodeLen] 
                        at ($(mid-8-1)-(1,0)$) {$\overline{x_n}$};
        
                \node (C1) [minimum width=\nodeWidth, minimum height=\nodeLen] 
                        at ($(mid-1-1)+(0,1)$) {$C_1$};
                \node (C2) [minimum width=\nodeWidth, minimum height=\nodeLen] 
                        at ($(mid-1-2)+(0,1)$) {$C_2$};
                \node (C3) [minimum width=\nodeWidth, minimum height=\nodeLen] 
                        at ($(mid-1-3)+(0,1)$) {$C_3$};
                \node (C4) [minimum width=\nodeWidth, minimum height=\nodeLen] 
                        at ($(mid-1-4)+(0,1)$) {$C_4$};
                \node (ldots1) [minimum width=\nodeWidth, minimum height=\nodeLen] 
                        at ($(mid-1-5)+(0,1)$) {$\ldots$};
                \node (C8) [minimum width=\nodeWidth, minimum height=\nodeLen] 
                        at ($(mid-1-6)+(0,1)$) {$C_8$};
                \node (ldots2) [minimum width=\nodeWidth, minimum height=\nodeLen] 
                        at ($(mid-1-7)+(0,1)$) {$\ldots$};
                \node (Cm) [minimum width=\nodeWidth, minimum height=\nodeLen] 
                        at ($(mid-1-8)+(0,1)$) {$C_m$};

                \node (n11) [minimum width=\nodeWidth, minimum height=\nodeLen] 
                        at ($(mid-1-2)$) {\tiny{$(1,1)$}};
                \node (n21) [minimum width=\nodeWidth, minimum height=\nodeLen] 
                        at ($(mid-2-2)$) {\tiny{$\overline{(1,1)}$}};
        
                \node (n12) [minimum width=\nodeWidth, minimum height=\nodeLen] 
                        at ($(mid-1-4)$) {\tiny{$\overline{(1,2)}$}};
                \node (n22) [minimum width=\nodeWidth, minimum height=\nodeLen] 
                        at ($(mid-2-4)$) {\tiny{$(1,2)$}};
        
                \node (n13) [minimum width=\nodeWidth, minimum height=\nodeLen] 
                        at ($(mid-1-6)$) {\tiny{$(1,3)$}};
                \node (n23) [minimum width=\nodeWidth, minimum height=\nodeLen] 
                        at ($(mid-2-6)$) {\tiny{$\overline{(1,3)}$}};
        
                \node (n31) [minimum width=\nodeWidth, minimum height=\nodeLen] 
                        at ($(mid-3-1)$) {\tiny{$\overline{(1,1)}$}};
                \node (n41) [minimum width=\nodeWidth, minimum height=\nodeLen] 
                        at ($(mid-4-1)$) {\tiny{$(1,1)$}};
        
                \node (n32) [minimum width=\nodeWidth, minimum height=\nodeLen] 
                        at ($(mid-3-2)$) {\tiny{$(2,2)$}};
                \node (n42) [minimum width=\nodeWidth, minimum height=\nodeLen] 
                        at ($(mid-4-2)$) {\tiny{$\overline{(2,2)}$}};
        
                \node (n33) [minimum width=\nodeWidth, minimum height=\nodeLen] 
                        at ($(mid-3-3)$) {\tiny{$\overline{(3,1)}$}};
                \node (n43) [minimum width=\nodeWidth, minimum height=\nodeLen] 
                        at ($(mid-4-3)$) {\tiny{$(3,1)$}};
        
                \node (n71) [minimum width=\nodeWidth, minimum height=\nodeLen] 
                        at ($(mid-7-2)$) {\tiny{$\overline{(1,3)}$}};
                \node (n81) [minimum width=\nodeWidth, minimum height=\nodeLen] 
                        at ($(mid-8-2)$) {\tiny{$(1,3)$}};

                \foreach \i in {5,...,6} {
                        \foreach \j in {1,...,8} {
                            \node (n\i\j) [minimum width=\nodeWidth, minimum height=\nodeLen] 
                        at ($(mid-\i-\j)+(0,0.25)$) {\tiny{$\vdots$}};
        
                        }
                    }
                \foreach \i in {1,...,8} {
                        \foreach \j in {5,7} {
                            \node (m\i\j) [minimum width=\nodeWidth, minimum height=\nodeLen] 
                        at ($(mid-\i-\j)$) {\tiny{$\ldots$}};
        
                        }
                    }
                \end{tikzpicture}
                \caption{\footnotesize{Relationship between variables and clauses dictating the nesting of posts. Empty cell means variable does not occur in clause(not all nonempty cells are shown in figure). A cell marked $(u,v)$ or  $\overline{(u,v)}$ indicates that it is the $u$-th occurrence of the variable in a clause and is the $v$-th variable of the clause.
                The bars on the tuple indicate the polarity of the variable occurrence. $C_2=x_1  \vee x_2 \vee \overline{x_n}$, $x_1$ occurs in $C_2,C_8$ and $\overline{x_1}$ occurs in $C_4$.}}
                \label{app:fig:varsAndClauses}
        \vspace{-\baselineskip}
\end{figure}

We will simply write this as $p=\tuple{h_1,\post,h_2,\post,\ldots,\post,h_n}$. In case $h_i=h_{i+1}=\ldots=h_j$ we will further shorten this to $\tuple{h_1,\post$,$h_2,\post,\ldots$, $h_{i-1}$, $\post^{j-i},h_j$, $\post,h_{j+1}$, $\post,\ldots,\post,h_n}$. 

\input{pingPongExample} 
Before explaining how we use this in our construction, let us look at a simple (and unrelated to the construction) example which shows what nested posting can achieve in Figure \ref{fig:bddHandlerQueue}. 
The figure shows the run of a program with three handlers $h_1,h_2,h_3$. Initially, in configuration $c_1$ (assume that some other handler has made these posts to $h_1$), the queue of $h_1$ contains 3 messages while the other two handlers have empty queues. For space reasons, the messages are written in short. The string $p_2(p_3(p_1(m_1)))$ is short for a message containing a single instruction $post(h_2,m'_1)$, which is at the head of the $h_1$ queue. Note that first $p_2$ indicates that the post instruction posts the message $m'_1$ to $h_2$. Here $m'_1$ is itself a message with a single instruction $post(h_3,m''_1)$ where $m''_1=post(h_1,m_1)$ for the message $m_1$. Our goal is to show how the `inner' messages $m_1,m_2$ and $m_3$ can be put into the queue of $h_1$ in certain orders but not in certain other orders.

 When $p_2(p_3(p_1(m_1)))$ is dequeued and executed, then $m'_1$ is posted and shows up in configuration $c_2$ as $p_3(p_1(m_1))$ in the queue of $h_2$. The other two messages remain in the queue of $h_1$ with the head of the queue now being $p_3(p_2(p_1(m_2)))$. The rightarrow indicates that it a one step transition from $c_1$ to $c_2$. The down arrow with a $*$ indicates that in multiple steps we go to $c_3$, the posting of the two messages in $h_1$ one by one to $h_3$. Focus on the inner most messages $m_2$ and $m_3$. Note that they are now `present' in $h_3$ (wrapped up in posts) in the same order as they were originally in $h_1$. Since the post sequence $p_2,p_3,p_1$ is the same for $m_2$ and $m_3$, they will always pass through different handlers in the same order as they were originally present in $h_1$. However, the post sequences are different for $h_1$. Skipping ahead to $c_5$, we see that $m_1$ is in $h_3$ while $m_2,m_3$ are in $h_2$. At this point, we have shown a sequence of transitions where $m_2$ and $m_3$ are posted first to $h_1$ followed by $m_1$. This results in the inner message order of $m_2,m_3,m_1$ in the queue of $h_1$. However, at configuration $c_5$, by executing $p_1(m_2)$, then $p_1(m_1)$ and then $p_1(m_3)$ we can instead obtain a configuration $c_7$ which results in $h_1$ being populated in the order $m_2,m_1,m_3$. To summarise, $m_1$ can be shuffled with $m_2$ and  $m_3$ in all possible ways, but $m_2$ must always be before $m_3$. In particular, this means we can never get $m_3,m_2,m_1$ in $h_1$. While this explanation has been with respect to program execution, the same logic can be lifted to traces.

Now let us see how the idea is used in our construction. 
For each row of the grid labelled by a literal $l$, we create a post sequence $p^l$. Consider the row labelled by $\overline{x_1}$ in the Figure \ref{app:fig:varsAndClauses}. The post sequence is the concatenation of 7 post sequences $p^{\overline{x_1}}=p^{\overline{x_1}}_1p^{\overline{x_1}}_2p^{\overline{x_1}}_3p^{\overline{x_1}}_4p^{\overline{x_1}}_5p^{\overline{x_1}}_6p^{\overline{x_1}}_7$ where $p^{\overline{x_1}}_2,p^{\overline{x_1}}_4,p^{\overline{x_1}}_6$ correspond to the cells marked $\overline{(1,1)}$, $(1,2)$ and $\overline{(1,3)}$ respectively, while the others correspond to the part of the row consisting of unmarked cells, with $p^{\overline{x_1}}_1$ for the part from the beginning till the first marked cell, etc. Note that in concatenation of post sequences, we have to add a get event in between appropriately. Each of the post sequences $p^{\overline{x_1}}_1$, $p^{\overline{x_1}}_3$, $p^{\overline{x_1}}_5$, $p^{\overline{x_1}}_7$ consists of a long sequence of posts to $h_V$ of length the number of unmarked cells in the segment they correspond to. For example $p^{\overline{x_1}}_1=p^{\overline{x_1}}_3=\tuple{h_V,\post,h_V}$ while $p^{\overline{x_1}}_5=\tuple{h_V,\post^3,h_V}$ since there are 3 empty cells in between (in the figure they are not explicitly shown, but rather by $\ldots$, but once can infer that the boxes corresponding to $C_5,C_6,C_7$ are empty along this row). The overall idea is that the post sequences are executed column by column. The post sequences of the empty cells simply `send to the back of queue' of  $h_V$ while the marked cells are responsible for shuffling the messages in the $h_{t_k}$ handlers which populate $h_W$ with an appropriate sequence of messages as we explain below.

We now describe the post sequences made in the marked cells. Consider the 6 marked cells corresponding to column $C_2$. Top to bottom, these are $p^{x_1}_2, p^{\overline{x_1}}_2, p^{x_2}_4, p^{\overline{x_2}}_4, p^{x_n}_2, p^{\overline{x_n}}_2$. Let us consider the post sequence for a cell labelled $(u,v)$ (resp. $\overline{(u,v)}$), indicating that it is the $u$-th occurrence of the variable in a clause and is the $v$-th variable of the clause, with the bar indicating whether the variable or its negation occurs in the clause. Suppose $u \neq 1$, then the post sequence is $\tuple{h_V,\post,h_{t_v},\post,h_V}$ for both $(u,v)$ as well as $\overline{(u,v)}$. If $u=1$ then the post sequence is $\tuple{h_V,\post^2,h_{t_v},\post,h_V}$ for $(u,v)$ but it is $\tuple{h_V,\post,h_{t_{v+3}},\post,h_{t_{v}},\post,h_V}$ for $\overline{(u,v)}$. 
For example,  $p^{x_n}_2$ which is marked $\overline{(1,3)}$ has the post sequence $\tuple{h_V,\post,h_{t_6},\post,h_{t_3},\post,h_V}$. Intuitively, the post sequences of the variable and its negation move to different handlers $h_{t_k}$ before coming back to the same handler iff a variable is occurring for the first time in a clause i.e., if $u=1$.

We modify each post sequence of a marked cell to post a message $m_{i,j,b}$ (corresponding to the occurrence of $x_i$ in $C_j$ in positive or negative form based on the value of the bit $b$) to $h_W$ just before its return to $h_V$. This is what we called \emph{message insertion}. For example, in $p^{x_n}_2$ we insert the events $e_2,e_3,e_4,e_5$ between the events $e_1$ and $e_6$ which are part of $p^{x_n}_2$ as follows:

\begin{figure}
\begin{tikzpicture}
        \foreach \i in {1,...,4} {
       \coordinate (E\i) at (({\i-1)  * \nodehGap},0); 
}
\coordinate (E5) at ({3  * \nodehGap},-\nodevGap);
\coordinate (E6) at ({1  * \nodehGap},-\nodevGap);
\coordinate (E7) at ({2  * \nodehGap},-\nodevGap);

\egetEvent{e_1}{E1}{h_{t_3}};
\epostEvent{e_2}{E2}{h_{t_3}}{h_W};
\egetEvent{e_3}{E3}{h_W};
\ewriteEvent{e_4}{E4}{h_W}{l^2_{n,1}}{1};
\ereadEvent{e_5}{E5}{h_W}{l^2_{n,2}};
\epostEvent{e_6}{E6}{h_{t_3}}{h_V};
\node[draw] (e_7) [minimum width=\nodeWidth, minimum height=\nodeLen] 
        at (E7) {$ Stage\ 2$}; 
\drawPo{e_1}{e_2};
\drawPb{e_2}{e_3};
\drawPo{e_3}{e_4};
\vdrawPo{e_4}{e_5};
\vdrawPo{e_2}{e_6};
\draw[dotted,->] (e_4.south) -- (e_7.north) node[midway, right] {$\rf$};
\draw[dotted,->] (e_7.east) -- (e_5.west) node[midway, right] {$\rf$};
\end{tikzpicture}
\caption{Inserting message into post sequence.}
\label{app:fig:msgInsertion}
\end{figure}
Consider the set $M_6$ of six messages posted to $h_W$ corresponding to $C_2$. The assignment to $x_1$ and $x_n$ are chosen by using different choices of $k$ of handlers $h_{t_k}$ for them, but the assignment to $x_2$ was already chosen when executing the post sequence for $C_1$. Hence $p^{x_2}_4, p^{\overline{x_2}}_4$ will both contain a single post to $h_{t_2}$ and the corresponding messages will be posted to $h_W$ in the order already chosen during the $C_1$ part. Note that two identical post sequences $p_1=p_2$ which start in some order in the queue of some handler $h$ will occupy the queue of subsequent handlers $h'$ of the post sequence in the same order due to queue semantics. Crucially, we prevent orderings of $M_6$ messages in $h_W$'s queue which do not correspond to consistent assignment of variables. However we allow all other possible reorderings of $M_6$ messages and this is essential for the verification in Stage 2, where only some of these reorderings may be allowed based on the partial order of events in a satisfiable clause i.e. one where not all red  $\hb$ arrows are present (see Figure \ref{app:fig:clauseGadget}).  

We now show the correctness of this construction. 

\subsection*{Correctness of the Construction}

\begin{lemma}
        The 3-BI-3SAT formula $\phi$ is satisfiable if and only if there exists a witnessing execution consistent with $\tau$.
\end{lemma}

\begin{proof}
     Note that in order to obtain a consistent trace, we will have to extend our constructed trace by specifying the $\eo$ and $\mo$ relations. These will induce a happens-before relation $\hb$ which needs to be acyclic. \\
\noindent\underline{$\phi$ is unsatisfiable.} For this direction, it suffices to look at the $\eo$ edges. Since $\phi$ is unsatisfiable, every assignment made to the variables must set some clause $C_j$ to false. In our example, consider the setting when $C_2=x_1  \vee x_2  \vee \overline{x_n}$ is false i.e. each of $x_1,x_2$ is set to false and $x_n$ is set to true. Recall that setting $x_n$ to true means that the message $\overline{m_{n,2,0}}$ has an $\eo$ edge leading to $m_{n,2,1}$. We have the following chain of relations: $e_{15} \; \rf \; e_{10} \; \po \; e_{11} \; \rf \; e_{16} \; \eo \; e_{17} \; \rf \; e_{12}$. This implies $e_{11} \xrightarrow{\hb} e_{12}$. Similarly we also have $e_{7} \xrightarrow{\hb} e_{8}$ and $e_{3} \xrightarrow{\hb} e_{4}$. Together with the existing relations, this creates a $\hb$ cycle, implying that this assignment of $\eo$ edges is inconsistent. Since we can find such a cycle for every assignment, we conclude that our constructed trace $\tau$ is inconsistent.   

\noindent\underline{$\phi$ is satisfiable.} Instead of specifying $\eo$ and $\mo$ separately, it is easier to specify a total order $\hb'$ on the set of all posts. 
 This clearly induces both the $\eo$ and the $\mo$ relations in a unique way. Intuitively, the $\hb'$ order executes parts of the post sequences in phases corresponding to the columns, with phase $m$ corresponding to the parts in column $m$. We complete all the posts in phase $m$ before moving on to phase $m+1$. 

  Let $P_\ell$ be the set of all posts made by handler $h_\ell$. We will assign the $\hb'$ edges in the following way:
 First we consider $P_V$. Moving column wise through the post sequences as in Figure \ref{app:fig:varsAndClauses} with $2n$ rows and $m$ columns, let $\tuple{p,q}$ refer to the cell in the $p$-th row and $q$-th column. Then we have the following total order on cells: $\tuple{p,q} < \tuple{p',q'}$ iff either $q < q'$ or ($q=q'$ and $p<p'$). Together with the nesting order of posts inside each cell, this is the $\hb'$ order on posts in $P_V$. The posts which remain are those in $P'= \bigcup_{k=1}^6 P_{t_k}$ which are all contained in the marked cells (note that $h_W$ does not make any posts and hence $P_W=\emptyset$). The $\hb'$ order between posts in $P_V$ and $P'$ is given by the nesting order. Looking at figure \ref{app:fig:msgInsertion}, an example of a post in $P'$ is $e_2$. Let $e,e'$ be posts in $P_V$ immediately before $e_1$ and immediately after $e_6$ in the post sequence $p^{x_n}$. Then we have $e \; \hb' \;  e_1  \; \hb' \; e_6 \; \hb' \; e'$.  
 It remains to choose the order between two posts made by different $h_{t_k}$ in $P'$ when choosing the assignment of a variable $x_i$. Note that there are 6 in total, we denote this set of posts $P_{x_i}$.

 In order to determine the order between two posts in $P_{x_i}$ we will look at what happens in Stage 2. We explain using the example in Figure \ref{app:fig:clauseGadget}. 

 Let $f$ be a satisfying assignment for $\phi$. Since $f$ satisfies $C_2$, this means that it  sets one of the literals in $C_2$ to true. Let us consider the case where $f(x_2)= f(x_n)= \true$ and $f(x_1)=\false$. This induces $e_{11} \xrightarrow{\hb} e_{12}$, $e_{7} \xrightarrow{\hb} e_{8}$ and $e_{4} \xrightarrow{\hb} e_{3}$. Overall we then get the following happens-before relation on the partial trace:\\
 $e_4 \xrightarrow{\hb} e_5 \xrightarrow{\hb} e_6 \xrightarrow{\hb} e_7 \xrightarrow{\hb} e_8 \xrightarrow{\hb} e_9 \xrightarrow{\hb} e_{10} \xrightarrow{\hb} e_{11} \xrightarrow{\hb} e_{12} \xrightarrow{\hb} e_{13} \xrightarrow{\hb} e_{14} \xrightarrow{\hb} e_1 \xrightarrow{\hb} e_2 \xrightarrow{\hb} e_3$.
 In this case, we already have a total ordering on the events. In the case where more than one literal of a clause is set to true by $f$, we will get a partial order for the $\hb$ relation. In all cases, we can choose a total order $<_2$ which extends this partial order. 
 
 Let us consider our example clause $C_2$ with assignment $f$. Since $x_1=\true$ and this is the first occurrence of  $x_1$ in any clause, the post sequences containing  $m_{1,2,0}$ and $m_{1,2,1}$ will first be sent to  $h_{t_1}$ and $h_{t_4}$ respectively. After this, they are both sent to $h_{t_1}$ from where they post the messages to $h_W$. This means that depending on which of the two post sequences is posted to $h_{t_1}$ just before they post to $h_W$, the value of $x_1$ is assigned. In the example since $x_1= \true$, this means that $m_{1,2,1}$ will be posted later than $m_{1,2,0}$.  After this, they are both sent back to $h_V$ in the same order in which they appeared in $h_{t_1}$ for the last time. The message sequences corresponding to $x_1$ and $\overline{x_1}$ still have two more messages to be posted to $h_W$. Since the remainder of the post sequences $p^{x_1}$ and $p^{\overline{x_1}}$ is identical, the order chosen between them is retained and the assignment to $x_1$ by the 4 other messages is consistent with the original choice. For example, $x_2$ is occurring for the second time when it appears in $C_2$. This means that the order between the messages  $m_{2,2,0}$ and $m_{2,2,1}$ has already been chosen during the posting of $m_{2,1,0}$ and $m_{2,1,1}$ executed previously. When  $m_{2,2,0}$ and $m_{2,2,1}$ are to be posted to $h_W$, the post sequence sends both to $h_{t_2}$, maintaining the prior order. 

 In general, the posts within $\bigcup_k h_{t_k}$ are completed before sending back to $h_V$ for each marked cell, and the execution continues with the sending of the next row with unmarked cell to the back of the queue in $h_V$. In this manner, all posts in phase $m$ are completed and we proceed to phase $m+1$.

 When all post events have been executed for all phases, the configuration consists of empty queues in all but the handler $h_W$, where the messages $m_{i,j,b}$ are placed in the order $<_2$ dictated by $f$. We can now execute the clause checking in Stage 2 i.e. the events in the handlers $h_{C_a}, h_{C_b}, h_{C_c}, h_{C_d}$ which are sandwiched with the writes of the messages in $h_W$. This completes the proof of correctness.   
\end{proof}

%% file: pingPongExample.tex
\def \tikzVoffsetNode{-3}
\def \tikzHoffsetArrow{3}
\def \tikzVoffsetArrow{-1.5}

\def \tikzHoffsetNode{6}

\def\ptwo{\textcolor{black}{$p_2$}}
\def \pone{$p_1$}
\def \pthree{$p_3$}

\begin{figure} 
    \begin{tikzpicture}[line width=1pt,framed,inner sep=1pt]

    \node[name=c1] at (0,0) { 
    \begin{tabular}{c|c| c}
    $h_1$ & $h_2$ & $h_3$\\ 
      \ptwo (\pthree (\pone( $m_1$))) & \quad & \quad\\
      \pthree (\ptwo (\pone( $m_2$))) & \quad & \quad\\
      \pthree (\ptwo (\pone( $m_3$))) & \quad & \quad
    \end{tabular}
    };
    \node[anchor=south east] at (c1.north west) {$c_1 =$};

    \node at (\tikzHoffsetArrow,0) {$\rightarrow$};

    \node[name=c2] at (\tikzHoffsetNode,0) { 
    \begin{tabular}{c|c|c}
    $h_1$ & $h_2$ & $h_3$\\ 
      \pthree (\ptwo (\pone( $m_2$))) & \pthree (\pone( $m_1$)) & \quad\\
      \pthree (\ptwo (\pone( $m_3$)))  & \quad & \quad\\
      \quad & \quad & \quad 
    \end{tabular}
    };
    \node[anchor=south east] at (c2.north west) {$c_2 =$};

    \node at ($(c2)+(0,\tikzVoffsetArrow)$) {\rotatebox{-90}{$\xrightarrow{\quad * \quad}$}};

    \node[name=c3] at ($(c2)+(0,\tikzVoffsetNode)$) { 
    \begin{tabular}{c|c|c}
   $h_1$ & $h_2$ & $h_3$\\ 
       \quad& \pthree (\pone( $m_1$)) & \ptwo (\pone( $m_2$)) \\
      \quad  & \quad & \ptwo (\pone( $m_3$))\\
      \quad & \quad & \quad 
    \end{tabular}
    };
    \node[anchor=south east] at (c3.north west) {$c_3 =$};

    \node at ($(c3)+(-\tikzHoffsetArrow,0)$) {$\xleftarrow{*}$};

    \node[name=c4] at ($(c3)+(-\tikzHoffsetNode,0)$) { 
    \begin{tabular}{c|c|c}
   $h_1$ & $h_2$ & $h_3$\\ 
       \quad& \pthree (\pone( $m_1$)) & \quad \\
      \quad  & \pone( $m_2$) & \quad\\
      \quad & \pone( $m_3$) & \quad 
    \end{tabular}
    };
    \node[anchor=south east] at (c4.north west) {$c_4 =$};
    \node at ($(c4)+(0,\tikzVoffsetArrow)$) {$\downarrow$};

    \node[name=c5] at ($(c4)+(0, \tikzVoffsetNode)$) { 
    \begin{tabular}{c|c|c}
   $h_1$ & $h_2$ & $h_3$\\ 
       \quad& \pone( $m_2$)& \pone( $m_1$) \\
      \quad  & \pone( $m_3$) & \quad\\
      \quad &  & \quad 
    \end{tabular}
    };
    \node[anchor=south east] at (c5.north west) {$c_5 =$};
    \node at ($(c5)+(\tikzHoffsetArrow,0)$) {$\xrightarrow{*}$};

    \node[name=c6] at ($(c5)+(\tikzHoffsetNode,0)$) { 
    \begin{tabular}{c|c|c}
   $h_1$ & $h_2$ & $h_3$\\ 
      $m_2$ & \quad & \quad \\
      $m_3$  & \quad & \quad\\
      $m_1$ & \quad & \quad  
    \end{tabular}
    };
    \node[anchor=south east] at (c6.north west) {$c_6 =$};

    \end{tikzpicture}
    \caption{Sorting messages via nested posts.}
    \label{fig:bddHandlerQueue}
  \end{figure}

%% file: appendix-polytimeProcedureQueues.tex
\section{Proof of Theorem \ref{thm:boundedHandlernoNestingQinP}}
\label{prf:thm:boundedHandlernoNestingQinP}
Consider a trace $\tau$ with $k$ handlers and with no nesting of posts. This implies that there is an initial message $m_i$ in each handler $h_i$ $1 \leq i \leq k$ such that all $\post$ events posted by $h_i$ occur in $m_i$. The set of all post instructions from a handler is totally ordered by the $\po$ relation since they are in the initial message. This implies a total $\mo$ order $\mo_{i,j}$ on the set $P(i,j)$ of all posts made by $h_i$ to $h_j$. Hence the trace already specifies $k$ many total orders on the set of all posts made \emph{to} a handler $h_j$. Due to queue semantics, this translates to $k$ many total orders on the messages corresponding to these posts. Note that the initial message $m_j$ occurs before each of the posted messages in $h_j$. Let $M_{i,j}=m_{i,j,1} \;\eo\; m_{i,j,2} \;\eo\; \ldots m_{i,j,l}$ be the set of messages corresponding to the posts $P_{i,j}$. Note that each message consists of sequence of events $e_1 \;\po\;  \ldots \;\po\; e_o$. 

We define a configuration $C$ as containing for each handler $h_j$:\\
(1) A pointer $s_{i,j}$ denoting an element of $M_{i,j}$ (or "start" if the first message in $M_{i,j}$ has not yet started, or "end" if all messages in $M_{i,j}$ have completed executing) for each $i$, and \\
(2) A pointer $r_j$ to an event which is either in the initial message or one of the messages in $M_{i,j}$ or "term".\\
The pointers in (1) indicate which messages have been executed, while (2) indicates the program pointer of the currently executing message in $h_j$.
 Note that if $r_j$ points to an event in the initial message $m_j$, then the pointer in (1) is set to "start".   

Each pointer can be stored in space $O(\log(n))$, hence the total amount of space required is $O(k^2\log(n))$. This implies that the total number of configurations is polynomial in $n$ and exponential in $k$. Let $\mathcal{C}$ denote the set of all configurations. We create a graph $G=(\mathcal{C},\mathcal{E})$ where the set of edges $\mathcal{E}$ is determined as follows:\\
Consider two configurations $C=(\{s_{i,j} \mid 1 \leq i,j  \leq k \}, s)$ and $C'=(\{s'_{i,j} \mid 1 \leq i,j  \leq k \}, s')$. There is an edge between $C$ and $C'$ iff one of the following holds:
Either $s'_{i,j}=s_{i,j}$ for all $j$, $r_j=r'_j$ for all $j$ except for a unique handler $h_q$ for which $r'_j$ points to the event which is the successor of $r_j$ under the  $\po$ order,\\
or there exists $j_0$ such that $r_j$ points to the last event of a message, $s'_{i,j}=s_{i,j}$ for all $j \neq j_0$, and $r'_j$ points to the first event of the successor message of $s_{i,j}$ under the $\eo$ order (or to "end" if all messages in $h_j$ have been executed). In the case where the initial message is finishing execution, the pointers for $s_{i,j}$ are all set to the first message in each of the $M_{i,j}$.  

The initial node $v_0$ of the graph sets all pointers $s_{i,j}$ to "start" and $r_j$ to the first event of every initial message. The last node $v_f$ sets all pointers $s_{i,j}$ to "end" and $r_j$ to "term". The trace is consistent iff there is a path from $v_0$ to $v_f$ in the graph.

%% file: app-experiments.tex
\section{Appendix for Section~\ref{sec:exp}}~\label{app:experiments}

\subsection*{Experimental Results for Traces generated via \Nidhugg}

We used the open-source model checker \Nidhugg to generate traces from event-driven programs. While \Nidhugg supports an event-driven execution model, it currently interprets asynchronous semantics using multisets rather than FIFO queues. As a result, some of the generated traces may be inconsistent under queue semantics. For each benchmark program, we randomly sampled traces from Nidhugg's output. These traces are guaranteed to satisfy multiset semantics, but may violate stricter queue-based consistency.

\input{table_exp-full}

\subsection*{Benchmarks}

\noindent We consider standard benchmark programs from prior works on event-driven programs~\cite{Kragl20}.
To demonstrate the subtleties of our setting, we have modified these examples to allow for multiple messages being posted to the same handler. 

\begin{itemize}
    \item \textbf{Consensus} is an example taken from Kragl et al.~\cite{Kragl20}. Here, the protocol outlines a straightforward broadcast consensus approach where $n$ nodes aim to reach agreement on a single value. Each node, labeled $i$, spawns two threads: one responsible for broadcasting, which sends the node's value to all other nodes, and another serving as an event handler that runs a collection process. This collection process gathers $n$ values and decides on the maximum, storing it as the node's final decision. As each node receives inputs from every other node, by the protocol's conclusion, all nodes converge on the same agreed-upon value.
    \item \textbf{Buyer}. This example involves $n$ buyer processes working together to purchase an item from a seller. Initially, one buyer obtains a price quote from the seller. The buyers then coordinate their individual contributions, and if their combined contributions meet or exceed the cost of the item, an order is placed. This benchmark is adapted from~\cite{Castro-pldi19}. 

    \item \textbf{Sparse-matrix} is a program that determines the count of non-zero elements in a sparse matrix with dimensions $m \times n$. The computation is carried out by splitting it into $n$ separate tasks, each sent as a message to different handlers. These handlers process their respective portions and then aggregate the results, ultimately combining the outputs from each task to produce the final count.
    \item \textbf{Message-loop} is a synthetic benchmark in which there are $n$ handlers. Whenever the $k$th handler receives a message, it increments a global message counter, and sends a new message to the $k+1$th handler, looping back from $n$ to $1$. Each such chain passes every handler $n$ times (for a total of $n^{2}$ messages). Initially, two chains are started; each at the first handler.
    \item \textbf{Counting} is a synthetic benchmark. Initially, each of $n$ handlers send a message to each other handler, and finally a message to itself. When a handler handles a message from a different handler $j$, it writes to a variable that it most recently received a message from handler $j$. This gets overwritten if it then handles a message from handler $k$. If a handler receives the message from itself, it reads the value of the shared variable, to see from which handler it has gotten the current value.   
\end{itemize}

%% file: table_exp-full.tex

\begin{table}[h]
	\vspace{-\baselineskip}
	\centering
	\resizebox{\textwidth}{!}{
		\vspace{-\baselineskip}
\begin{tabular}{| l | c | c  | c | c | c | c | c | c | c | c |}
\hline
  \multirow{3}{*}{\textbf{Benchmark}} &
\multirow{3}{*}{\# E} &
\multirow{3}{*}{\# M} &
\multirow{3}{*}{\# H} &
\multirow{3}{*}{\# T} &
\multicolumn{3}{c|}{Algorithm 1} & \multicolumn{3}{c|}{Algorithm 2} \\
& & & & & \# Consistent & \# T/O & Time & \# Consistent & \# T/O & Time \\
& & & & & traces & traces & in sec. & traces & traces & in sec. \\
	\hline
Buyers (2) & 88 & 6 & 2 & 2 & 2 & 0 & 0.0071 & 2 & 0 & 0.0295 \\
Buyers (4) & 166 & 10 & 2 & 10 & 10 & 0 & 14.3448 & 10 & 0 & 0.0358 \\
Buyers (8) & 322 & 18 & 2 & 10 & 0 & 10 & - & 10 & 0 & 0.0603 \\
ChangRoberts (2) & 207 & 11 & 3 & 2 & 2 & 0 & 0.0930 & 2 & 0 & 0.0356 \\
ChangRoberts (4) & 353 & 17 & 5 & 10 & 4 & 6 & 36.1441 & 10 & 0 & 0.0432 \\
ChangRoberts (8) & 737 & 33 & 9 & 10 & 0 & 10 & - & 10 & 0 & 0.0687 \\
Consensus (2) & 221 & 9 & 3 & 4 & 2 & 0 & 0.1624 & 2 & 0 & 0.0338 \\
Consensus (4) & 677 & 25 & 5 & 10 & 0 & 10 & - & 1 & 0 & 0.0702 \\
Consensus (8) & 2333 & 81 & 9 & 10 & 0 & 10 & - & 1 & 0 & 2.3321 \\
Counting (2) & 129 & 9 & 3 & 4 & 3 & 0 & 0.0269 & 4 & 0 & 0.0329 \\
Counting (4) & 443 & 25 & 5 & 10 & 0 & 10 & - & 10 & 0 & 0.0596 \\
Counting (8) & 1647 & 81 & 9 & 10 & 0 & 10 & - & 10 & 0 & 1.1833 \\
MessageLoop (2) & 268 & 23 & 3 & 10 & 1 & 0 & 0.3115 & 1 & 0 & 0.1172 \\
MessageLoop (4) & 954 & 73 & 5 & 10 & 1 & 0 & 72.2307 & 1 & 0 & 4.7470 \\
MessageLoop (8) & 3670 & 269 & 9 & 10 & 0 & 10 & - & 0 & 10 & - \\
SparseMat (2) & 231 & 7 & 3 & 7 & 7 & 0 & 0.1507 & 7 & 0 & 0.0375 \\
SparseMat (4) & 427 & 11 & 3 & 10 & 6 & 4 & 64.3954 & 10 & 0 & 0.0485 \\
SparseMat (8) & 819 & 19 & 3 & 10 & 0 & 10 & - & 10 & 0 & 0.1141 \\
	\hline
	\end{tabular}
	}
	\caption{\footnotesize{Experimental results for benchmark programs collected from droidracer. The field \# T denotes the number of traces. The traces can differ in size (events \# E), messages \# M, handlers \# H), and the field contains the maximum of its traces. The field \# Consistent traces denotes the number of these traces for which the implementation reports the existence of a satisfying execution. The field \# T/O traces denotes the number of traces for which our tool timed out (with a timeout of 120s). For any remaining traces, the tool concludes inconsistency. The time fields represent the average runtime for the traces that did not time out. A value of - indicates that the corresponding algorithm timed out on every trace.}
	}
	\label{tab:nidhugg_results}
	\vspace{-\baselineskip}
\end{table}

%% file: ed-consistency-arxiv.bbl
\def\Nst#1{$#1^{st}$}\def\Nnd#1{$#1^{nd}$}\def\Nrd#1{$#1^{rd}$}\def\Nth#1{$#1^{th}$}
\begin{thebibliography}{10}
\providecommand{\url}[1]{\texttt{#1}}
\providecommand{\urlprefix}{URL }
\providecommand{\doi}[1]{https://doi.org/#1}

\bibitem{DBLP:journals/acta/AbdullaAAJLS17}
Abdulla, P.A., Aronis, S., Atig, M.F., Jonsson, B., Leonardsson, C., Sagonas,
  K.: Stateless model checking for {TSO} and {PSO}. Acta Inf.  \textbf{54}(8),
  789--818 (2017)

\bibitem{DBLP:conf/popl/AbdullaAJS14}
Abdulla, P.A., Aronis, S., Jonsson, B., Sagonas, K.: Optimal dynamic partial
  order reduction. In: {POPL}. pp. 373--384. {ACM} (2014)

\bibitem{ATVA2023}
Abdulla, P.A., Atig, M.F., Bonneland, F.M., Das, S., Lang, M., Jonsson, B.,
  Sagonas, K.: Tailoring stateless model checking for event-driven
  multi-threaded programs. In: Automated Technology for Verification and
  Analysis - 21st International Symposium, {ATVA} 2023, Proceedings. Lecture
  Notes in Computer Science, vol. 14216. Springer (2023)

\bibitem{Abdulla2019}
Abdulla, P.A., Atig, M.F., Jonsson, B., L{\aa}ng, M., Ngo, T.P., Sagonas, K.:
  Optimal stateless model checking for reads-from equivalence under sequential
  consistency. Proc. {ACM} Program. Lang.  \textbf{3}({OOPSLA}),  150:1--150:29
  (2019)

\bibitem{DBLP:conf/cav/AbdullaAJL16}
Abdulla, P.A., Atig, M.F., Jonsson, B., Leonardsson, C.: Stateless model
  checking for {POWER}. In: {CAV} {(2)}. Lecture Notes in Computer Science,
  vol.~9780, pp. 134--156. Springer (2016)

\bibitem{DBLP:journals/pacmpl/AbdullaAJN18}
Abdulla, P.A., Atig, M.F., Jonsson, B., Ngo, T.P.: Optimal stateless model
  checking under the release-acquire semantics. {PACMPL}  \textbf{2}({OOPSLA}),
   135:1--135:29 (2018)

\bibitem{microsoft}
kernel-mode~driver architecture., M.I.W.:
  http://msdn.microsoft.com/en-us/library/windows/hardware/ff557560(v=vs.85).aspx

\bibitem{bouajjaniVerifyingRobustnessEventDriven2017}
Bouajjani, A., Emmi, M., Enea, C., Ozkan, B.K., Tasiran, S.: Verifying
  robustness of event-driven asynchronous programs against concurrency. In:
  {ESOP}. Lecture Notes in Computer Science, vol. 10201, pp. 170--200. Springer
  (2017)

\bibitem{Bouajjani17}
Bouajjani, A., Enea, C., Guerraoui, R., Hamza, J.: On verifying causal
  consistency. In: {POPL}. pp. 626--638. {ACM} (2017)

\bibitem{Castro-pldi19}
Castro{-}Perez, D., Hu, R., Jongmans, S., Ng, N., Yoshida, N.: Distributed
  programming using role-parametric session types in go: statically-typed
  endpoint apis for dynamically-instantiated communication structures. Proc.
  {ACM} Program. Lang.  \textbf{3}({POPL}),  29:1--29:30 (2019)

\bibitem{Chakraborty2024}
Chakraborty, S., Krishna, S.N., Mathur, U., Pavlogiannis, A.: How hard is
  weak-memory testing? Proc. {ACM} Program. Lang.  \textbf{8}({POPL}),
  1978--2009 (2024)

\bibitem{Concuerror:ICST13}
Christakis, M., Gotovos, A., Sagonas, K.: Systematic testing for detecting
  concurrency errors in erlang programs. In: {ICST}. pp. 154--163. {IEEE}
  Computer Society (2013)

\bibitem{CGMP:partialorder}
Clarke, E.M., Grumberg, O., Minea, M., Peled, D.A.: State space reduction using
  partial order techniques. Int. J. Softw. Tools Technol. Transf.
  \textbf{2}(3),  279--287 (1999)

\bibitem{CunninghamK05}
Cunningham, R., Kohler, E.: Making events less slippery with eel. In: HotOS.
  USENIX Association (2005)

\bibitem{Dabek:event-driven-02}
Dabek, F., Zeldovich, N., Kaashoek, M.F., Mazi{\`{e}}res, D., Morris, R.:
  Event-driven programming for robust software. In: {ACM} {SIGOPS} European
  Workshop. pp. 186--189. {ACM} (2002)

\bibitem{P:pldi13}
Desai, A., Gupta, V., Jackson, E.K., Qadeer, S., Rajamani, S.K., Zufferey, D.:
  {P:} safe asynchronous event-driven programming. In: {PLDI}. pp. 321--332.
  {ACM} (2013)

\bibitem{fischerTasksLanguageSupport2007}
Fischer, J., Majumdar, R., Millstein, T.D.: Tasks: language support for
  event-driven programming. In: {PEPM}. pp. 134--143. {ACM} (2007)

\bibitem{gantyAnalyzingRealTimeEventDriven2009}
Ganty, P., Majumdar, R.: Analyzing real-time event-driven programs. In:
  {FORMATS}. Lecture Notes in Computer Science, vol.~5813, pp. 164--178.
  Springer (2009)

\bibitem{GayLBWBC03}
Gay, D., Levis, P., von Behren, J.R., Welsh, M., Brewer, E.A., Culler, D.E.:
  The nesc language: A holistic approach to networked embedded systems. In:
  Cytron, R., Gupta, R. (eds.) PLDI. pp. 1--11. ACM (2003)

\bibitem{apple}
central~dispatch (GCD)~reference., A.C.I.G.: http://developer.apple.com/library
  / mac/\#documentation/ performance/ reference/ gcd\_libdispatch\_ref/
  reference/ reference.html

\bibitem{GibbonsK97}
Gibbons, P.B., Korach, E.: Testing shared memories. {SIAM} J. Comput.
  \textbf{26}(4),  1208--1244 (1997)

\bibitem{Godefroid:thesis}
Godefroid, P.: Partial-Order Methods for the Verification of Concurrent
  Systems: An Approach to the State-Explosion Problem. Ph.D. thesis, University
  of {Li{\`e}ge} (1996). \doi{10.1007/3-540-60761-7},
  \url{http://www.springer.com/gp/book/9783540607618}, also, volume~1032 of
  {LNCS}, Springer.

\bibitem{Godefroid:popl97}
Godefroid, P.: Model checking for programming languages using verisoft. In:
  {POPL}. pp. 174--186. {ACM} Press (1997)

\bibitem{Godefroid:verisoft-journal}
Godefroid, P.: Software model checking: The verisoft approach. Formal Methods
  Syst. Des.  \textbf{26}(2),  77--101 (2005)

\bibitem{GoHaJa:heartbeat}
Godefroid, P., Hanmer, R.S., Jagadeesan, L.J.: Model checking without a model:
  An analysis of the heart-beat monitor of a telephone switch using verisoft.
  In: {ISSTA}. pp. 124--133. {ACM} (1998)

\bibitem{HillSWHCP00}
Hill, J.L., Szewczyk, R., Woo, A., Hollar, S., Culler, D.E., Pister, K.S.J.:
  System architecture directions for networked sensors. In: {ASPLOS}. pp.
  93--104. {ACM} Press (2000)

\bibitem{Event-DrivenSMC-OOPSLA-15}
Jensen, C.S., M{\o}ller, A., Raychev, V., Dimitrov, D., Vechev, M.T.: Stateless
  model checking of event-driven applications. In: {OOPSLA}. pp. 57--73. {ACM}
  (2015)

\bibitem{KLSV:popl18}
Kokologiannakis, M., Lahav, O., Sagonas, K., Vafeiadis, V.: Effective stateless
  model checking for {C/C++} concurrency. Proc. {ACM} Program. Lang.
  \textbf{2}({POPL}),  17:1--17:32 (2018)

\bibitem{Kokologiannakis2022}
Kokologiannakis, M., Marmanis, I., Gladstein, V., Vafeiadis, V.: Truly
  stateless, optimal dynamic partial order reduction. Proc. {ACM} Program.
  Lang.  \textbf{6}({POPL}),  1--28 (2022)

\bibitem{KoSa:spin17}
Kokologiannakis, M., Sagonas, K.: Stateless model checking of the linux
  kernel's hierarchical read-copy-update (tree {RCU}). In: {SPIN}. pp.
  172--181. {ACM} (2017)

\bibitem{GenMC-CAV-21}
Kokologiannakis, M., Vafeiadis, V.: {GenMC}: {A} model checker for weak memory
  models. In: {CAV} {(1)}. Lecture Notes in Computer Science, vol. 12759, pp.
  427--440. Springer (2021)

\bibitem{Kragl20}
Kragl, B., Enea, C., Henzinger, T.A., Mutluergil, S.O., Qadeer, S.: Inductive
  sequentialization of asynchronous programs. In: {PLDI}. pp. 227--242. {ACM}
  (2020)

\bibitem{libasync}
LIBASYNC: http://pdos.csail.mit.edu/6.824-2004/async/

\bibitem{libevent}
LIBEVENT: http://monkey.org/~provos/libevent/

\bibitem{EDpedagogy2021}
Lukkarinen, A., Malmi, L., Haaranen, L.: Event-driven programming in
  programming education: {A} mapping review. {ACM} Trans. Comput. Educ.
  \textbf{21}(1),  1:1--1:31 (2021)

\bibitem{maiyaPartialOrderReduction2016}
Maiya, P., Gupta, R., Kanade, A., Majumdar, R.: Partial order reduction for
  event-driven multi-threaded programs. In: {TACAS}. Lecture Notes in Computer
  Science, vol.~9636, pp. 680--697. Springer (2016)

\bibitem{Maiya2014DroidRacerArtefact}
Maiya, P., Kanade, A., Majumdar, R.: Droidracer tested apps repository.
  \url{https://bitbucket.org/iiscseal/droidracer-related-files/src/master/pldi-2014-tested-apps/}
  (2014), artefact accompanying the PLDI 2014 paper "Race Detection for Android
  Applications"

\bibitem{Maiya:pldi14}
Maiya, P., Kanade, A., Majumdar, R.: Race detection for android applications.
  In: {PLDI}. pp. 316--325. {ACM} (2014)

\bibitem{Mazieres01}
Mazi{\`{e}}res, D.: A toolkit for user-level file systems. In: Park, Y. (ed.)
  Proceedings of the General Track: 2001 {USENIX} Annual Technical Conference.
  pp. 261--274. {USENIX} (Jun 2001),
  \url{http://www.usenix.org/publications/library/proceedings/usenix01/mazieres.html}

\bibitem{mednieks2012programming}
Mednieks, Z., Dornin, L., Meike, G.B., Nakamura, M.: Programming {Android}.
  "O'Reilly Media, Inc." (2012)

\bibitem{deMoura2008Z3}
de~Moura, L.M., Bj{\o}rner, N.S.: {Z3:} an efficient {SMT} solver. In: {TACAS}.
  Lecture Notes in Computer Science, vol.~4963, pp. 337--340. Springer (2008)

\bibitem{MQBBNN:chess}
Musuvathi, M., Qadeer, S., Ball, T., Basler, G., Nainar, P.A., Neamtiu, I.:
  Finding and reproducing heisenbugs in concurrent programs. In: {OSDI}. pp.
  267--280. {USENIX} Association (2008)

\bibitem{NoDe:toplas16}
Norris, B., Demsky, B.: A practical approach for model checking {C/C++11} code.
  {ACM} Trans. Program. Lang. Syst.  \textbf{38}(3),  10:1--10:51 (2016)

\bibitem{Peled:representatives}
Peled, D.A.: All from one, one for all: on model checking using
  representatives. In: {CAV}. Lecture Notes in Computer Science, vol.~697, pp.
  409--423. Springer (1993)

\bibitem{mace}
Project, T.M.: http://mace.ucsd.edu

\bibitem{raychevEffectiveRaceDetection2013}
Raychev, V., Vechev, M.T., Sridharan, M.: Effective race detection for
  event-driven programs. In: {OOPSLA}. pp. 151--166. {ACM} (2013)

\bibitem{Tovey84}
Tovey, C.A.: A simplified np-complete satisfiability problem. Discret. Appl.
  Math.  \textbf{8}(1),  85--89 (1984)

\bibitem{Tunc2023}
Tun{\c{c}}, H.C., Abdulla, P.A., Chakraborty, S., Krishna, S., Mathur, U.,
  Pavlogiannis, A.: Optimal reads-from consistency checking for c11-style
  memory models. Proc. {ACM} Program. Lang.  \textbf{7}({PLDI}),  761--785
  (2023)

\bibitem{Valmari:reduced:state-space}
Valmari, A.: Stubborn sets for reduced state space generation. In: Applications
  and Theory of Petri Nets. Lecture Notes in Computer Science, vol.~483, pp.
  491--515. Springer (1989)

\bibitem{DBLP:conf/pldi/ZhangKW15}
Zhang, N., Kusano, M., Wang, C.: Dynamic partial order reduction for relaxed
  memory models. In: {PLDI}. pp. 250--259. {ACM} (2015)

\end{thebibliography}
